\documentclass[a4paper,12pt,authoryear]{elsarticle}
\usepackage[utf8]{inputenc}
\usepackage[english]{babel}
\usepackage{amsmath,amssymb,amsfonts,amsthm,textcomp,natbib}
\graphicspath{{./images/}}
\usepackage{listings}
\usepackage{tabularx}
\usepackage{fullpage}
\usepackage{enumitem}
\newtheorem{theorem}{Theorem}
\newtheorem{definition}{Definition}

\newtheorem{lemma}{Lemma}

\makeatletter
\@addtoreset{proofpart}{theorem}
\makeatother
\usepackage{tikz}
\usetikzlibrary{positioning,arrows,shapes,shadows,snakes}
\usetikzlibrary{intersections,decorations.markings,calc,positioning}

\usepackage{soul}

\makeatletter
\newcommand\RedeclareMathOperator{%
  \@ifstar{\def\rmo@s{m}\rmo@redeclare}{\def\rmo@s{o}\rmo@redeclare}%
}
\newcommand\rmo@redeclare[2]{%
  \begingroup \escapechar\m@ne\xdef\@gtempa{{\string#1}}\endgroup
  \expandafter\@ifundefined\@gtempa
     {\@latex@error{\noexpand#1undefined}\@ehc}%
     \relax
  \expandafter\rmo@declmathop\rmo@s{#1}{#2}}
\newcommand\rmo@declmathop[3]{%
  \DeclareRobustCommand{#2}{\qopname\newmcodes@#1{#3}}%
}
\@onlypreamble\RedeclareMathOperator
\makeatother

\newcommand{\SE}[2][]{\mathbf{SE}_{#2}^{#1}}
\newcommand{\NW}[2][]{\mathbf{NW}^{#1}_{#2}}

\DeclareMathOperator{\cgeq}{\succcurlyeq}
\DeclareMathOperator{\cleq}{\preccurlyeq}
\DeclareMathOperator{\cgt}{\succ}
\DeclareMathOperator{\clt}{\prec}

\RedeclareMathOperator{\S}{\mathbf{S}}

\begin{document}
\title{Axiomatization of the Choquet integral for 2-dimensional heterogeneous product sets}
\author[t1]{Mikhail Timonin}
\ead[t1]{m.timonin@qmul.ac.uk}
\address[t1]{Queen Mary University of London \\ Mile End Road London, UK, E1 4NS \\  +442078826139  }
\date{}

%\tableofcontents

\begin{abstract}
  We prove a representation theorem for the Choquet integral model. The preference relation is defined on a two-dimensional heterogeneous
  product set $X = X_1 \times X_2$ where elements of $X_1$ and $X_2$ are not necessarily comparable with each other. However, making such
  comparisons in a meaningful way is necessary for the construction of the Choquet integral (and any rank-dependent model). We construct the
  representation, study its uniqueness properties, and look at applications in multicriteria decision analysis, state-dependent utility
  theory, and social choice. Previous axiomatizations of this model, developed for decision making under uncertainty, relied heavily on the
  notion of comonotocity and that of a ``constant act''. However, that requires $X$ to have a special structure, namely, all factors of this
  set must be identical. Our characterization does not assume commensurateness of criteria a priori, so defining comonotonicity becomes
  impossible.
\end{abstract}


\maketitle
\section{Introduction}
\label{sec:introduction}

Rank-dependent models appeared in axiomatic decision theory in reply to the criticism of Savage's postulates of rationality
\citep{savage1954foundations}. The renowned Ellsberg paradox \citep{ellsberg1961risk} has shown that people can violate Savage's axioms and
still consider their behaviour rational. First models accounting for the so-called uncertainty aversion observed in this paradox appeared in
the 1980s, in the works of \citet{quiggin1982theory} and others (see \citep{wakker1991additive-RO} for a review). One particular
generalization of the expected utility model (EU) characterized by \citet{schmeidler1989subjective} is the Choquet expected utility (CEU),
where probability is replaced by a non-additive set function (called capacity) and integration is performed using the Choquet integral.

Since Schmeidler's paper, various versions of the same model have been characterized in the literature (e.g. \citep{gilboa1987expected,
  wakker1991additive}). CEU has gained some momentum in both theoretical and applied economic literature, being used mainly for analysis of
problems involving Knightian uncertainty. At the same time, rank-dependent models, in particular the Choquet integral, were adopted in
multiattribute utility theory (MAUT) \citep{keeney1993decisions}. Here the integral gained popularity due to the tractability of
non-additive measures in this context (see \citep{grabisch2008decade} for a review). The model permitted various preferential phenomena,
such as criteria interaction, which were impossible to reflect in the traditional additive models.

The connection between MAUT and decision making under uncertainty has been known for a long time. In the case when the number of states is
finite, which is assumed hereafter, states can be associated with criteria. Accordingly, acts correspond to multicriteria
alternatives. Finally, the sets of outcomes at each state can be associated with the sets of criteria values. However, this last transition
is not quite trivial. It is commonly assumed that the set of outcomes is the same in each state of the world
\citep{savage1954foundations,schmeidler1989subjective}. In multicriteria decision making the opposite is true. Indeed, consider preferences
of consumers choosing cars. Each car is characterized by a number of features (criteria), such as colour, maximal speed, fuel consumption,
comfort, etc. Apparently, sets of values taken by each criterion can be completely different from those of the others. In such context the
ranking stage of rank-dependent models, which in decision under uncertainty involves comparing outcomes attained at various states, would
amount to comparing colours to the level of fuel consumption, and maximal speed to comfort.

Indeed, the traditional additive model \citep{debreu1959topological,krantz1971foundation} only implies meaningful comparability of
\textit{units} between goods in the bundle, but not of their absolute levels. However, in rank-dependent models such comparability seems to
be a necessary condition. This paper develops a characterization of the Choquet integral for two-dimensional sets with comparability
(commensurateness) of the criteria not assumed a priori.

Let $X = X_1 \times X_2$ be a (heterogeneous) product set and $\succcurlyeq$ a binary relation defined on this set. In multiattribute
utility theory, elements of the set $X$ are interpreted as alternatives characterized by two criteria taking values from sets $X_1$ and
$X_2$. In decision making under uncertainty, the factors of the set $X$ usually correspond to outcomes in various states of the world, and an
additional restriction $X_1 = X_2 = Y$ is being made. Thus in CEU, the set $X$ is homogeneous, i.e. $X = Y^n$.

Previous axiomatizations of the Choquet integral have been given for this special case of $X = Y^n$ (see
\citep{kobberling2003preference} for a review of approaches) and its variant $X = \mathbb{R}^n$ (see \citep{grabisch2008decade} for a
review). Another approach using conditions on the utility functions was proposed in \citep{labreuche2012axiom}. A conjoint axiomatization of
the Choquet integral for the case of a general $X$ was an open problem in the literature. One related result that should be mentioned is the
recent conjoint axiomatization of another non-additive integral, the Sugeno integral (\citep{greco2004axiomatic,bouyssou2009conjoint}).

The crucial difference between our result and previous axiomatizations is that the notions of ``comonotonicity'' and ``constant
act'' are no longer available in the heterogeneous case. Recall that two acts are called comonotonic in CEU if their outcomes have the
same ordering. A constant act is plainly an act having the same outcome in every state of the world. Apparently, since criteria sets $X_1$
and $X_2$ in our model can be completely disjoint, neither of the notions can be used anymore due to the fact that there does not exist a
meaningful built-in order between elements of sets $X_1$ and $X_2$. New axioms and proof techniques must be introduced to
deal with this complication.

The paper is organized as follows. Section \ref{sec:choquet-integral} defines the Choquet integral and looks at its properties. Section
\ref{sec:axioms} states and explains the axioms. Section \ref{sec:repr-theor} gives the representation theorem. Section
\ref{sec:discussion-results} discusses the main result and its economic interpretations. The proof of the theorem is presented in the
Appendix; in particular, necessity of axioms is discussed in section \ref{sec:neccessity-axioms}.

% \section{Introduction}
% \label{sec:introduction}

% % \section{Definitions}
% % \label{sec:definitions}

% % \begin{definition}
% %   $bq$ is south-east from $ap$ ($bq \se ap$) iff $ap \cgeq aq$ for all $a$
% %   and $bp \cgeq ap$ for all $p$.
% % \end{definition}

% % \begin{definition}
% %   $ap$ is north-west from $bq$ ($ap \nw bq$) iff $ap \cgeq aq$ for all $a$
% %   and $bp \cgeq ap$ for all $p$.
% % \end{definition}

% % \begin{lemma}
% %   \label{lm:2}
% %   For any $ap \sim bq$, $ap SE bq$ or $ap NW bq$
% % \end{lemma}
% % \begin{proof}
% %   Follows from weak separability.
% % \end{proof}

\section{Choquet integral in MAUT}
\label{sec:choquet-integral}

Let $ N = \{1, 2 \} $ be a set (of criteria) and $ 2^N $ its power set.

\begin{definition}
  \emph{Capacity (non-additive measure, fuzzy measure)} is a set function
  $\nu:2^N\rightarrow \mathbb{R}_+$ such that:
  \begin{enumerate}
  \item $ \nu(\varnothing)=0 $;
  \item $ A \subseteq B \Rightarrow \nu(A)\leq\nu(B), \ \forall A,B \in 2^N $.
  \end{enumerate}
  In this paper, it is also assumed that capacities are normalized, i.e.\ $ \nu(N) = 1 $.
\end{definition}

\begin{definition}
  The \emph{Choquet integral} of a function $ f:N \rightarrow \mathbb{R} $ with respect to a capacity $ \nu $ is defined as
  \begin{equation*}
    C(\nu,f) = \int \limits_{0}^{\infty} \nu( \{ i \in N \colon f(i) \geq r \})dr + \int \limits_{-\infty}^{0} [\nu( \{ i \in N \colon f(i)
    \geq r \}) - 1]dr
  \end{equation*}
\end{definition}
Denoting the range of $f:N \rightarrow \mathbb{R}$ as $(f_1,\ldots,f_n)$, the definition can be expressed as:
   \begin{equation*}
    C(\nu,(f_1,\ldots,f_n)) =
    \sum \limits_{i=1}^{n} (f_{(i)}-f_{(i-1)})\nu(\{ j \in N \colon f_j \geq  f_{(i)} \} )
  \end{equation*}
  where $f_{(1)},\ldots,f_{(n)}$ is a permutation of $f_1,\ldots,f_n$ such that $f_{(1)} \leq f_{(2)} \leq \cdots \leq
  f_{(n)}$, and $f_{(0)}=0$.

\subsection{The model}
\label{sec:model}

Let $\cgeq$ be a binary relation on the set $X = X_1 \times X_2$. $\cgt, \clt, \cleq, \sim, \not \sim$ are defined in the usual way. % In
% MCDA, elements of set $X$ are interpreted as alternatives characterized by criteria from the set $N = \{1,2\}$. Sets $X_1$ and $X_2$ contain
% criteria values for criteria 1 and 2 respectively. 
We say that $\cgeq$ can be represented by a Choquet integral, if there exists a capacity
$\nu$ and functions $f_1: X_1 \rightarrow \mathbb{R}$ and $f_2: X_2 \rightarrow \mathbb{R}$, called value functions, such that:
\begin{equation*}
  x \cgeq y \iff C(\nu,(f_1(x_1),f_2(x_2)) \geq C(\nu,(f_1(y_1),f_2(y_2)).
\end{equation*}

As seen in the definition of the Choquet integral, its calculation involves comparison of the $f_i$'s to each other. It is not immediately
obvious how this operation can have any meaning in the MAUT context. It is well-known that comparing \textit{levels} of value functions for
various attributes is meaningless in the additive model \citep{krantz1971foundation} (recall that the origin of each value function can be
changed independently). In the homogeneous case $X = Y^n$ this problem is readily solved, as we have a single set of outcomes $Y$ (in the
context of decision making under uncertainty). The required order is either assumed as given \citep{wakker1991additive-RO} or is readily
derived from the ordering of the constant acts $(\alpha, \ldots, \alpha)$ \citep{wakker1991additive}. Since there is a single outcome set,
we also have a single value (utility) function $U:Y \rightarrow \mathbb{R}$, and thus comparing $U(y_1)$ to $U(y_2)$ is perfectly sensible, since $U$
represents the order on the set $Y$. None of these methods can be readily applied in the heterogeneous case.

\subsection{Properties of the Choquet integral}
\label{sec:prop-choq-integr}

Below are given some important properties of the Choquet integral:
\begin{enumerate}
\item Functions $f:N \rightarrow \mathbb{R}$ and $g:N \rightarrow \mathbb{R}$ are comonotonic if for no $i,j \in N$ we have $f(i) > f(j)$ and
  $g(i) < g(j)$. For all comonotonic $f$ the Choquet integral reduces to the Lebesgue integral. In the finite case, the integral is
  accordingly reduced to a weighted sum.
\item Particular cases of the Choquet integral (e.g. \citep{grabisch2008decade}). % Assume $N = \{1,2\}$.
  \begin{itemize}
  \item If $\nu(\{1\}) = \nu(\{2\}) = 1$, then $C(\nu,(f_1,f_2)) = \max(f_1,f_2)$.
  \item If $\nu(\{1\}) = \nu(\{2\}) = 0$, then $C(\nu,(f_1,f_2)) = \min(f_1,f_2)$.
  \item If $\nu(\{1\}) + \nu(\{2\}) = 1$, then $C(\nu,(f_1,f_2)) = \nu(\{1\})f_1 + \nu(\{2\})f_2$
  \end{itemize}
\end{enumerate}

Property 1 states that the set $X$ can be separated into subsets corresponding to particular orderings of the value functions. In the case
of two criteria there are only two such sets: $\{x \in X : f_1(x_1) \geq f_2(x_2)\}$ and $\{x \in X: f_2(x_2) \geq f_1(x_1) \}$. Since the
integral on each of the sets is reduced to a weighted sum, i.e. an additive representation, we should expect many of the axioms of the
additive conjoint model to be valid on this subsets. This is the intuition behind several of the axioms given in the following section.

\section{Axioms}
\label{sec:axioms}

\begin{definition}
  A relation $\cgeq$ on $X_1 \times X_2$ satisfies \emph{triple cancellation}, provided that for every $a,b,c,d \in X_1$ and $p,q,r,s \in X_2$ , if
  $ap \cleq bq, ar \cgeq bs$, and $cp \cgeq dq$, then $cr \cgeq ds$.
\end{definition}

\begin{definition}
  A relation $\cgeq$ on $X_1 \times X_2$ is \emph{independent}, iff for $a,b \in X_1, ap \cgeq bp$ for some $p \in X_2$ implies that $aq \cgeq bq$ for every $q \in X_2$; and,
  for $p,q \in X_2, ap \cgeq aq$ for some $a \in X_1$ implies that $bp \cgeq bq$ for every $b \in X_1$.
\end{definition}

\begin{enumerate}[label=\textbf{A{\arabic*}.}]
    \item $\cgeq$ is a weak order. \label{sec:am0}
    \item Weak separability - for any $a_ip_j, b_ip_j \in X$ such that $ a_ip_j \cgt b_ip_j$, we have $a_iq_j \cgeq b_iq_j$ for all
      $q \in X_j$, for $i,j \in \{1,2\}$. \label{sec:am0} 
\end{enumerate}
The separability condition is weaker than the one normally used. \footnote{The condition first appeared in \citep{bliss1975capital}, and in
  this form in \citep{mak1984notes}} In fact, it only rules out a reversal of strict preference. Note, that the
condition implies that for any $a,b \in X_1$ either $ap \cgeq bp$ or $bp \cgeq ap$ for all $p \in X_2$ (symmetrically for the second
coordinate). Apparently, transitivity also holds: if $ap \cgeq bp$ for all $p \in X_2$ and $bp \cgeq cp$ for all $p \in X_2$, then $ap \cgeq
cp$ for all $p \in X_2$. This allows to introduce the following weak orders:
\begin{definition}
  For all $a, b \in X_1$ define $\cgeq_1$ as $a \cgeq_1 b \iff ap \cgeq bp$ for all $p \in X_2$. Define $\cgeq_2$ symmetrically. 
\end{definition}
\begin{definition}
  We call $a \in X_1 $ \emph{minimal} if $b \cgeq_1 a$ for all $b \in X_1$, and \emph{maximal} if $a \cgeq_1 b$ for all $b \in
  X_1$. Symmetric definitions hold for $X_2$.
\end{definition}
\begin{definition}
  For any $z \in X$ define $\SE{z} = \{x \colon x \in X, x_1 \cgeq_1 z_1 ~\text{and}~ z_2 \cgeq_2 x_2\}$, and $\NW{z} = \{x \colon x \in X,
  x_2 \cgeq_2 z_2 ~\text{and}~ z_1 \cgeq_1
  x_1\}$.
\end{definition}
The ``rectangular''cones $\SE{z}$ and $\NW{z}$ play a significant role in the sequel. 
\begin{enumerate}[resume,label=\textbf{A{\arabic*}.}]
    \item For any $z \in X$, triple cancellation holds either on $\SE{z}$ or on $\NW{z}$.  \label{sec:am2}
\end{enumerate}
The axiom says that the set $X$ can be covered by ``rectangular'' cones, such that triple cancellation holds within each cone. We will call
such cones ``3C-cones''. The axiom effectively divides $X$ into subsets, defined as follows.
\begin{definition} We say that
  \begin{itemize}
  \item $x \in \SE{}$ if:
    \begin{itemize}
    \item There exists $z \in X$ such that $z_1$  is not maximal and $z_2$ is not minimal, triple cancellation holds on $\SE{z}$, and $x \in
      \SE{z}$, or
    \item $x_1$ is maximal or $x_2$ is minimal and for no $y \in \SE{x} \setminus x$ triple cancellation holds on $\NW{x}$;
    \end{itemize}
  \item $x \in \NW{}$ if:
    \begin{itemize}
    \item There exists $z \in X$ such that $z_1$  is not maximal and $z_2$ is not minimal, triple cancellation holds on $\SE{z}$, and $x \in
      \SE{z}$, or
    \item $x_1$ is minimal or $x_2$ is maximal and for no $y \in \NW{x} \setminus x$ triple cancellation holds on $\SE{x}$.
    \end{itemize}
  \end{itemize}
Define also $\Theta = \NW{} \cap \SE{}$. 
\end{definition}
\begin{definition}
  We say that $i \in N$ is \emph{essential} on $A \subset X$ if there exist $x_ix_j, y_ix_j \in A$, $i,j \in N$, such that $x_ix_j \cgt y_ix_j$.
\end{definition}
Essentiality of coordinates is discussed in detail in Section \ref{sec:essentiality}.
\begin{enumerate}[resume,label=\textbf{A{\arabic*}.}]
  \item  Whenever $ ap \cleq bq, ar \cgeq bs, cp \cgeq dq $, we have that $cr \cgeq ds$, provided that either:
    \begin{enumerate}[label={\alph*})]
    \item $ap, bq, ar, bs, cp, dq, cr, ds \in \NW{} (\SE{})$, or;
    \item $ap, bq, ar, bs \in \NW{}$ and $i=2$ is essential on $\NW{}$ and $cp, dq, cr, ds \in \SE{}$ or vice versa, or;
    \item $ap, bq, cp, dq \in \NW{}$ and $i=1$ is essential on $\NW{}$ and $cp, dq, cr, ds \in \SE{}$ or vice versa.
    \end{enumerate} \label{sec:am3}
  \end{enumerate}
  Informally, the meaning of the axiom is that ordering between preference differences (``intervals'') is preserved irrespective of the
  ``measuring rods'' used to measure them. However, contrary to the additive case this does not hold on all $X$, but only when either points
  involved in all four relations lie in a single 3C-cone, or points involved in two relations lie in one 3C-cone and those involved in the
  other two in another.
\begin{enumerate}[resume,label=\textbf{A{\arabic*}.}]
\item Whenever $ap \cleq bq, cp \cgeq dq$ and $ay_0 \sim x_0\pi(a), by_0 \sim x_0\pi(b), cy_1 \sim x_1\pi(c), dy_1 \sim x_1\pi(d)$, and also
  $e\pi(a) \cgeq f\pi(b)$, we have $e\pi(c) \cgeq f\pi(d)$, for all $ap,bq,cp,dq \in \NW{}$ or $\SE{}$ provided coordinate $i=1$ is essential on the
  subset which contains these points, $ay_0, by_0, cy_1, dy_1 \in \NW{}$ or $\SE{}$, $x_0\pi(a),x_0\pi(b),x_1\pi(c),x_1\pi(d) \in \NW{}$ or
  $\SE{}$ provided coordinate $i=2$ is essential on the subset which contains these points, $e\pi(a), f\pi(b), e\pi(c), f\pi(d) \in \NW{}$ or
  $\SE{}$. Same condition holds for the other dimension symmetrically.
    \label{sec:am5}
  \end{enumerate}
  The formal statement of \textbf{A5} is rather complicated, but it simply means that the ordering of the intervals is preserved across
  dimensions. Together with \textbf{A4} the conditions are similar to Wakker's trade-off consistency condition
  \citep{wakker1991additive-RO}. The axiom bears even stronger similarity to Axiom 5 (compatibility) from section 8.2.6 of
  \citep{krantz1971foundation}. Roughly speaking, it says that if the interval between $c$ and $d$ is larger than that between $a$ and $b$,
  then projecting these intervals onto another dimension by means of the equivalence relations must leave this order unchanged. We
  additionally require the comparison of intervals and projection to be consistent - meaning that quadruples of points in each part of the
  statement lie in the same 3C-cone. Another version of this axiom, which is going to be used frequently in the proofs, is formulated in
  terms of standard sequences in Lemma \ref{lm:A5}.
\begin{enumerate}[resume,label=\textbf{A{\arabic*}.}]
  \item Bi-independence : Let $ap,bp,cp,dp \in \SE{}(\NW{})$ and $ap \cgt bp$. If for some $q \in X_2$ also exist $cq \cgt dq$, then $cp
    \cgt dp$. Symmetric condition holds for the second coordinate.
  \end{enumerate}
  This axiom is similar to ``strong monotonicity'' in \citep{wakker1991additive-RO}. We analyze its necessity and the intuition behind it in
  section \ref{sec:essentiality}.
\begin{enumerate}[resume,label=\textbf{A{\arabic*}.}]
  \item Both coordinates are essential on $X$. 
  \item Restricted solvability : if $x_ia_j \cgeq y \cgeq x_ic_j$, then there exists $b: x_ib_j \sim y$, for $i,j \in \{1,2\}$.
    \label{am6}
  \item Archimedean axiom: for every $z \in \NW{}(\SE{})$ every bounded standard sequence contained in $\NW{z}(\SE{z})$ is finite.
    \label{am7}
  \end{enumerate}

\paragraph{Structural assumption.} For no $a,b \in X_1$ we have $ap \sim bp$ for all $p \in X_2$. Similarly, for no $p,q \in X_2$ we have $ap
\sim aq$ for all $a \in X_1$. If such points exist, say $ap \sim bp$ for all $p \in X_2$, then we can build the representation for a set
$X_1' \times X_2$ where $X_1' = X_1 \setminus a$, and later extend it to $X$ by setting $f_1(a) = f_1(b)$.
\label{sec:struct-assumpt}
%
% \linebreak
% Note that the structural assumption together with bi-independence implies that if $ap \cgeq dp$ for all $p \in X_2$ and exist $cq, dq \in
% \SE{}(\NW{})$ such that $cq \cgt dq$, then for all $p \in X_2$ such that $ap, bp \in \SE{}(\NW{})$ it holds $ap \cgt bp$. This is very
% similar to ``strong monotonicity'' in \citep{wakker1991additive-RO}. 
\paragraph{$X$ is order dense.} Whenever $x \cgt y$ there exists $z$ such that $x \cgt z \cgt y$. From this and restricted solvability
immediately follows that $\cgeq_i$ is order dense as well, in other words, whenever $a_ip_j \cgt b_ip_j$ there exists $c \in X_i$ such that
$a_ip_j \cgt c_ip_j \cgt b_ip_j$, for $i,j \in N$.

\paragraph{"Closedness" of $\SE{}$ and $\NW{}$.} Finally, we \emph{extend} the set $X$ as follows. Whenever exist $ap \not \in \NW{}$ and
$bp \not \in \SE{}$, there exist also $cp \in \Theta$. Similarly, whenever exist $ap \not \in \NW{}$ and
$aq \not \in \SE{}$, there exist also $ar \in \Theta$.

\subsection{Discussion of axioms}
\label{sec:discussion-axioms}

Roughly speaking, for two dimensional sets the Choquet integral can be characterized by saying that $X$ is divided into two subsets such
that $\cgeq$ on each of them has an additive representation, while the intersection of these subsets (in the representation) is the line
$\{x : f_1(x_1) = f_2(x_2)\}$. In the previous characterizations locating these subsets was straightforward, as they are nothing else but
the comonotonic subsets of $X$. In this paper we take a different approach. Instead, we state that $X$ can be separated into two subsets
without imposing any additional constraints on their location and then use additional axioms to characterize the intersection of these
subsets and to show that it is mapped to the line $\{x : f_1(x_1) = f_2(x_2)\}$.

Our axioms aim to reflect the main properties of the Choquet integral. The first one is that the set $X$ can be divided into subsets, such
that within every such subset the preference relation can be represented by an additive function. The axiom (\textbf{A3}) we introduce is
similar to the ``2-graded'' condition previously used for characterizing of MIN/MAX and the Sugeno integral
(\citep{greco2004axiomatic,bouyssou2009conjoint}). At every point $z \in X$ it is possible to build two ``rectangular cones'': $\{x : x_1
\cgeq_1 z_1 ~\text{and}~ z_2 \cgeq_2 x_2\}$, and $\{x : x_2 \cgeq_2 z_2 ~\text{and}~ z_1 \cgeq_1 x_1\}$. The axiom
states that triple cancellation must then hold on at least one of these cones.

The second property is that the additive representations on different subsets are interrelated, in particular trade-offs between criteria
values are consistent across subsets both within the same dimension and for different ones. This is reflected by two axioms (\textbf{A4,
  A5}), similar to the ones used in \citep{wakker1991additive} and \citep{krantz1971foundation} (section 8.2). One, roughly speaking, states
that triple cancellation holds across cones, while the other says that the ordering of intervals on any dimension must be preserved when
they are projected onto another dimension by means of equivalence relations. 

These axioms are complemented by a new condition called bi-independence (\textbf{A6}) and weak separability (\textbf{A2})
% \citep{bouyssou2009conjoint}
- which together reflect the monotonicity property of the integral.

Standard essentiality,``comonotonic'' Archimedean axiom and restricted solvability (\textbf{A7,A8,A9}) complete the list. Finally, $\cgeq$
is supposed to be a weak order, and $X$ is order dense (\textbf{A1}).

Our most important axioms - \textbf{A3,A4,A5,A6}, are not only sufficient, but also necessary. Necessity and detailed analysis of
\textbf{A6} is given in Section \ref{sec:essentiality}, necessity of \textbf{A4} and \textbf{A5} is proved in Section
\ref{sec:neccessity-axioms}, whereas necessity of \textbf{A3} is immediate (in the representation one of the regions $\NW{z}$ and $\SE{z}$
is necessarily contained in a comonotonic subset of $\mathbb{R}^2$). Necessity of some of the remaining axioms is well-known
\cite{wakker1991additive-RO, bouyssou2004preferences}.

\section{Representation theorem}
\label{sec:repr-theor}

\begin{theorem}
  \label{theo:repr}
  Let $\cgeq$ be an order on $X$ and let $X$ be order dense and the structural assumption hold. Then, if axioms \textbf{A1}-\textbf{A9} are
  satisfied, there exists a capacity $\nu$ and value functions $f_1:X_1 \rightarrow \mathbb{R}, f_2:X_2 \rightarrow
  \mathbb{R}$, such that $\cgeq$ can be represented by the Choquet integral:
  \begin{equation}
    \label{eq:repr}
    x \cgeq y \iff C(\nu,(f_1(x_1),f_2(x_2))) \geq C(\nu,(f_1(y_1),f_2(y_2))),
  \end{equation}
  for all $x,y \in X$.  Moreover, $\nu$ is determined uniquely and value functions have the following uniqueness properties:
  \begin{enumerate}
  \item If $\nu(\{1\}) + \nu(\{2\}) = 1$, then for any functions $g_1:X_1 \rightarrow \mathbb{R}, g_2:X_2 \rightarrow \mathbb{R}$
    such that (\ref{eq:repr}) holds with $f_i$ substituted by $g_i$, we have $f_i (x_i) = \alpha g_i (x_i) + \beta_i$ for some $\alpha > 0$.
  \item If $\nu(\{1\}) \in (0,1)$ and $\nu(\{2\}) \in (0,1)$ and $\nu(\{1\}) + \nu(\{2\}) \neq 1$, then for any functions $g_1:X_1
    \rightarrow \mathbb{R}, g_2:X_2 \rightarrow \mathbb{R}$ such that (\ref{eq:repr}) holds with $f_i$ substituted by $g_i$, we have $f_i
    (x_i) = \alpha g_i (x_i) + \beta$ for some $\alpha > 0$.
  \item If $\nu(\{2\}) \in (0,1)$, $\nu(\{1\}) \in \{0,1\}$, then for any functions $g_1:X_1 \rightarrow \mathbb{R}, g_2:X_2
    \rightarrow \mathbb{R}$ such that (\ref{eq:repr}) holds with $f_i$ substituted by $g_i$, we have : 
    \begin{itemize}
    \item $f_i (x_i) = \alpha g_i (x_i) + \beta$, for all $x$ such that $f_1(x_1) < \max f_2(x_2)$ and $f_2(x_2) > \min f_1(x_1)$;
    \item $f_i(x_i) = \psi_i(g_i(x_i))$ where $\psi_i$ is an increasing function, otherwise.
    \end{itemize}
  \item If $\nu(\{2\}) \in \{0,1\}$, $\nu(\{1\}) \in (0,1)$, then for any functions $g_1:X_1
    \rightarrow \mathbb{R}, g_2:X_2 \rightarrow \mathbb{R}$ such that (\ref{eq:repr}) holds with $f_i$ substituted by $g_i$, we have :
    \begin{itemize}
    \item $f_i (x_i) = \alpha g_i (x_i) + \beta$, for all $x$ such that $f_2(x_2) < \max f_1(x_1)$ and $f_1(x_1) > \min f_2(x_2)$;
    \item $f_i(x_i) = \psi_i(g_i(x_i))$ where $\psi_i$ is an increasing function, otherwise.
    \end{itemize}
  \item If $\nu(\{1\}) = \nu(\{2\}) = 0$ or $\nu(\{1\}) = \nu(\{2\}) = 1$, then for any functions $g_1:X_1 \rightarrow \mathbb{R}, g_2:X_2
    \rightarrow \mathbb{R}$ such that (\ref{eq:repr}) holds with $f_i$ substituted by $g_i$, we have : $f_i(x_i) = \psi_i(g_i(x_i))$ where
    $\psi_i$ are increasing functions such that $f_1(x_1) = f_2(x_2) \iff g_1(x_1) = g_2(x_2)$.
  \end{enumerate}
\end{theorem}

\subsection{Uniqueness properties imply commensurateness}
\label{sec:uniq-impl-comm}

As the uniqueness part of Theorem \ref{theo:repr} states, unless $\cgeq$ can be represented by an additive functional on all of $X$ (Case
1), the representation implies commensurateness of \textit{levels} of utility functions defined on different factors of the product
set. Indeed, we have that if $f_1(x_1) \geq f_2(x_2)$ in one representation, then necessarily $g_1(x_1) \geq g_2(x_2)$ in another one. This
is a much stronger uniqueness result in comparison to the traditional additive models. In Section \ref{sec:discussion-results} we discuss
some economic implications of this.

\subsection{Extension to n dimensions}
\label{sec:extens-n-crit}

This paper provides a characterization of the Choquet integral for two-dimensional sets, which allows to have simpler proofs. We believe
that an extension to $n$ dimensions is mostly a technical task. Axiom \textbf{A3} would be separated into two conditions. One is similar to
the current \textbf{A3}, and holds for any pair of dimensions with the remaining coordinates fixed, and the other is acyclicity of the
absence of additivity on the n-criteria counterparts of regions $\NW{z}$ and $\SE{z}$ in between pairs of coordinates. Just as in the
present paper, stronger uniqueness would be due to the lack of additivity. The remaining differences are technical.

% It seems that the approach can also be used to characterize a version of Gilboa-Schmeidler's ``min of means''(or MEU) model
% \citep{gilboa1989maxmin} in the heterogeneous context. In fact, for two-dimensional sets MEU is a strict subset of CEU corresponding to
% submodular capacities ($\nu(A \cup B) \geq \nu(A) + \nu(B) - \nu(A \cap B)$ for all $A,B \subset N$).

\section{Applications}
\label{sec:discussion-results}

\subsection{Multicriteria decision analysis}
\label{sec:interpretation-maut}

In multicriteria context our result implies that the decision maker constructs a one-to-one mapping between elements of criteria sets (their
subsets to be precise). Some authors interpret this by saying that criteria elements sharing the same utility values present the same level
of ``satisfaction'' for the decision maker \citep{grabisch2008decade}. Technically, such statements are meaningful, in the sense that
permissible scale transformations do not render them ambiguous or incorrect, unless the representation is additive. However, the substance
of statements like ``$x_1$ on criterion 1 is at least as good as $x_2$ on criterion 2'' (which would correspond to $f_1(x_1) \geq f_2(x_2)$)
is not easy to grasp. Perhaps, it is possible to think about workers performing various tasks within a single project. From the perspective
of a project manager, achievements of various workers, serving as criteria in this example, can be level-comparable despite being physically
different, if the project has global milestones (i.e. scale) which are mapped to certain personal milestones for every involved person. The
novelty of our characterization is that this scale is not given a priori. Instead, we only observe preferences of the project manager and
infer all corresponding mappings from them.

% Another observation is related to rank-dependent models in general. It can be said that the decision maker has at max two (in the two
% dimensional case) regions where his trade-off rates, measured in utility units, are constant. Our results show that the set of alternatives
% and the frontier between the two regions can be given in quite sparse and general terms, however, provided certain additional properties
% (mainly \textbf{A4} and \textbf{A5}) are met, it is possible to construct a well-defined rank-dependent model of preferences defined of such
% sets.
% In fact, value functions $f_1$ and $f_2$ can be seen to form a so-called Guttman scale (or a biorder)
% \citep{guttman1944basis,doignon1984realizable}. This is usually formalized as follows. $X_1$ is taken to be a set of individuals, and $X_2$ a

\subsection{Rank-dependent state-dependent utility }
\label{sec:state-depend-util}

State-dependent utility concept is evoked when the nature of the state itself is of significance to the decision-maker. One popular example
is healthcare, where various outcomes can have major effects on the personal value of the insurance premium \citep{karni1985decision}. In
the state-dependent context preferences of the decision maker are given by a binary relation on a set $X = Y^2$. However, unlike in CEU,
there exist two (by the number of states) utility functions $U_1 :Y \rightarrow \mathbb{R}$ and $U_2 :Y \rightarrow \mathbb{R}$, so that $x
\cgeq y \iff C(\nu,(U_1(x_1),U_2(x_2))) \geq C(\nu,(U_1(x_1),U_2(x_2)))$, where $x,y$ are acts, and $x_1,x_2,y_1,y_2$ are respectively
outcomes of $x$ and $y$ in each of the states $(1,2)$. We note that regions $\{z :z \in X, U_1(z_1) \geq U_2(z_2)\}$ and $\{z :z \in X,
U_1(z_1) \leq U_2(z_2)\}$ do not necessarily correspond to comonotonic regions of $X$ anymore. Constant acts also do not have any special
value in such context, since $U_1(x_1)$ and $U_2(x_1)$ are not necessarily equal. However, our characterization allows to construct this
model. Moreover, as a further generalization, in our framework sets of outcomes in every state can be completely disjoint as well.

\subsection{Social choice}
\label{sec:interpr-soci-choice}

If we think of the set $N$ as of a society with two agents, then $X$ is the set of all possible welfare distributions. Moreover, contrary to
the classical scenario, agents could be receiving completely different goods, for example $X_1$ might correspond to healthcare options,
whereas $X_2$ to various educational possibilities. In this case it is not a trivial task to build a correspondence between different
options across agents. Our result basically states that, provided the preferences of the social planner abide by the axioms given in Section
\ref{sec:axioms}, the decisions are made as if the social planner has associated cardinal utilities with the outcomes of each agent which
are \textit{unit} and \textit{level} comparable (cardinal fully comparable or CFC in terms of \citet{roberts1980interpersonal}). Such
approach is not conventional in social choice problems, where the global (social) ordering is usually not considered as given. Instead, the
conditions are normally given on individual utility functions and the ``aggregating'' functional that is used to derive the global
ordering. However, one of the important questions in social choice literature is that of the interpersonal utility comparability and whether
it is justifiable to assume it or not. Our results show that in case the global ordering of alternatives made by the society (or the social
planner) satisfy certain conditions, it is in principle possible to have individual preferences represented by utility functions that are
not only unit but also level comparable between each other.

\section*{Appendix}
\label{cha:appendix}

Subsequent sections are organized as follows. Section \ref{sec:proof-plan} contains a brief sketch of the proof. Section
\ref{sec:essentiality} investigates monotonicity properties, Sections \ref{sec:technical-lemmas},
\ref{sec:build-addit-repr}-~\ref{sec:repr-choq-integr} contain the main body of the proof: construction of capacity and value functions,
Section \ref{sec:uniqueness} analyses uniqueness of the obtained representation, finally, necessity of the axioms is shown in Section
\ref{sec:neccessity-axioms}.

\renewcommand{\thesection}{A.\arabic{section}}
\setcounter{section}{0}

\section{Proof sketch}
\label{sec:proof-plan}

\begin{enumerate}
\item Define extreme points of $\SE{}$ and $\NW{}$ and temporarily remove them from $X$ (Section \ref{sec:build-addit-repr}).
\item Take any point $z$, show that there exists an additive representation for $\cgeq$ on $\NW{z}$ if $z \in \NW{}$ or $\SE{z}$ if
  otherwise (Section \ref{sec:build-addit-repr}).
\item Having built additive representations for $\cgeq$ on $\SE{z_1}$ and $\SE{z_2}$, show that there exists an 
  additive representation on $\SE{z_1} \cup \SE{z_2}$ (Section \ref{sec:build-addit-repr}).
\item Cover all $\SE{}$ with 3C-cones, and show that the joint representation, call it $V^{SE}$, can also be extended to cover all $\SE{}$ (Section \ref{sec:build-addit-repr}). 
\item Perform steps 2 and 3 for \textbf{NW} and obtain $V^{NW}$ (Section \ref{sec:build-addit-repr}).
\item Align and scale $V^{SE}$ and $V^{NW}$ such that $V^{SE}_1 = V^{NW}_1$ on the common domain, and $V^{SE}_2 = \lambda V^{NW}_2$ on their
  common domain (Section \ref{sec:join-togeth-addit}).
\item Pick two points $r^0, r^1$ from $\Theta$ and set $r^0$ as a common zero. Set $V^{SE}_1(r^1_1) = 1$ and define $\phi_1(x_1) =
  V^{SE}(x_1), \phi_2(x_2) = V^{SE}(x_2)/V^{SE}(r^1_2)$ (Section \ref{sec:join-togeth-addit}).
\item Show that for all $x \in X$ we have $\phi_1(x_1) = \phi_2(x_2)$ iff  $x \in \Theta$ unless $\cgeq$ can be represented by an additive
  functional on all $X$ (Section \ref{sec:0-set}).
\item Representations now are $\phi_1 + k\phi_2$ on $\SE{}$ and $\phi_1 + \lambda k \phi_2$ on $\NW{}$ (Section \ref{sec:0-set}).
\item Rescale so that ``weights'' sum up to one : $\frac{1}{1+k}\phi_1 + \frac{k}{1+k}\phi_2$, $\frac{1}{1+\lambda k}\phi_1 + \frac{\lambda k}
  {1+\lambda k}\phi_2$ (Section \ref{sec:0-set}).
\item Extend the representation to the extreme points (Section \ref{sec:extreme-points}).
\item Show that $\cgeq$ can be represented on $X$ by these two representations (Section \ref{sec:vnw-vse-represent}).
\item Show that $\cgeq$ can be represented by the Choquet integral (Section \ref{sec:repr-choq-integr}). 
\end{enumerate}

\section{Technical lemmas}
\label{sec:technical-lemmas}

\begin{lemma}
  \label{lm:27}
  If $\cgeq$ satisfies triple cancellation then it is independent.
\end{lemma}

\begin{proof}
  $ap \cleq ap, aq \cgeq aq, ap \cgeq bp \Rightarrow aq \cgeq bq.$
\end{proof}

\begin{lemma}
  \label{lm:5}
  $X = \SE{} \cup \NW{}$.
\end{lemma}
\begin{proof}
  Assume $x \not \in \SE{}, x \not \in \NW{}$. First assume that $x$ is such that its coordinates are not maximal or minimal. Then, there
  does not exist $z$ such that $x \in \SE{z}$ and triple cancellation holds on $\SE{z}$. At the same time, there does not exist $z$ such
  that $x \in \NW{z}$ and triple cancellation holds on $\NW{z}$. This implies that triple cancellation does not hold on $\SE{x}$ or $\NW{x}$
  (otherwise we could have taken $z = x$). This violates \textbf{A3}.

Now assume $x_1$ is maximal. $x \not \in \SE{}$ implies that there exists $y \in \SE{x}$ such that triple cancellation holds on
$\NW{y}$. But then $x \in \NW{}$, a contradiction. Other cases are symmetrical.  
\end{proof}

\begin{lemma}
  \label{lm:A5}
  Axiom \textbf{A5} implies the following condition. Let $\{g^{(i)}_1 : g^{(i)}_1y_0 \sim g^{(i+1)}_1y_1, g^{(i)}_1 \in X_1, i \in N \}$ and
  $\{h^{(i)}_2 : x_0h^{(i)}_2 \sim x_1h^{(i+1)}_2, h^{(i)}_2 \in X_x, i \in N \}$ be two standard sequences, each entirely contained in
  $\NW{}$ or $\SE{}$. Assume also, that there exist $z_1, z_2 \in X$, $p,q \in X_2, a,b \in X_1$ such that $g^{(i)}_1p, g^{(i)}_1q \in \NW{}$ or
  $\SE{}$, and $ah^{(i)}_2, bh^{(i)}_2 \in \NW{}$ or $\SE{}$ for all $i$, and $g^{(i)}_1p \sim bh^{(i)}_2$ and $g^{(i+1)}_1p \sim
  bh^{(i+1)}_2$, then for all $i \in N$, $g^{(i)}_1p \sim bh^{(i)}_2$.
\end{lemma}
\begin{proof}
  The proof is very similar to the one from \cite{krantz1971foundation} (Lemma 5 in section 8.3.1).  Assume wlog that $\{g^{(i)}_1:
  g^{(i)}_1y_0 \sim g^{(i+1)}_1y_1\}$ is an increasing standard sequence on $X_1$, which is entirely in $\SE{}$, whereas $\{h^{(i)}_1:
  x_0h^{(i)}_1 \sim x_1h^{(i+1)}_1\}$ is an increasing standard sequence on $X_2$, and lies entirely in $\NW{}$. Assume also for some $k$ it
  holds $g_1^{(k)}y_0 \sim x_0h_2^{(k)}, g_1^{(k+1)}y_0 \sim x_0h_2^{(k+1)}$. We need to show that $g_1^{(i)}y_0 \sim x_0h_2^{(i)}$ for all
  $i$. We will show that $g_1^{(k+2)}y_0 \sim x_0h_2^{(k+2)}$ from which everything holds by induction.

  Assume $x_0h_2^{(k+2)} \cgt g_1^{(k+2)}y_0$. Since the sequences are increasing, by restricted solvability exists $g \in X_2$ such that
  $g_1^{(k+2)}y_0 \sim x_0g$. By \textbf{A5'}, $g^{(k)}_1y_0 \sim g^{(k+1)}_1y_1, g^{(k+1)}_1y_0 \sim g^{(k+2)}_1y_1, x_0h^{(k)}_2 \sim
  x_1h^{(k+1)}_2$ imply $x_0h^{(k+1)}_2 \sim x_1g$. By definition of $\{h^{(i)}_2\}$, $x_0h^{(k+1)}_2 \sim x_1h^{(k+2)}_2$. Thus,
  $x_1h^{(k+2)}_2 \sim x_1g$ and by independence $x_0h^{(k+2)}_2 \sim x_0g$, hence also $g^{(k+2)}_1y_0 \sim x_0h^{(k+2)}$, a
  contradiction. The case $x_0h_2^{(k+2)} \clt g_1^{(k+2)}y_0$ is symmetrical. Showing that $g_1^{(k-1)}y_0 \sim x_0h_2^{(k-1)}$ can be done
  in a similar fashion.
\end{proof}

\begin{lemma}
  \label{lm:34}
  The following statements hold:
  \begin{itemize}
  \item If $\NW{}$($\SE{}$) has only $X_2$ essential, then for all $x \in \NW{}$($\SE{}$) there exists $y_2 \in X_2$ such that $x_1y_2 \in \Theta$.
  \item If $\NW{}$($\SE{}$) has only $X_1$ essential, then for all $x \in \NW{}$($\SE{}$) there exists $y_1 \in X_1$ such that $y_1x_2 \in \Theta$.
  \end{itemize}
\end{lemma}
\begin{proof}
  Immediately follows from the structural assumption and closedness of $\NW{}$($\SE{}$).
\end{proof}

\section{Essentiality and monotonicity}
\label{sec:essentiality}

In what follows the essentiality of coordinates within various $\SE{z}(\NW{z})$ is critical. The central mechanism to guarantee consistency
in number of essential coordinates within various 3C-cones is bi-independence which is closely related to comonotonic strong monotonicity of
\cite{wakker1989additive}. 

% \citet{wakker1991additive-RO}, starting with an order on $Y$, imposes a different version of strong monotonicity directly, stating that
% $\alpha \cgeq' \beta \iff x_{-i} \alpha \cgeq x_{-i}\beta$. However, in fact this condition is not necessary for additive representations on
% comonotonic cones as can be seen in the following example.

% Assume there exists a comonotonic cone $C^n \subset Y^n$, and there exists a maximal $y^{\max} \in Y$. Let all conditions from
% \citep{wakker1991additive-RO} hold. Then we can build an additive representation on $C^n$. Now extend it in the following way: let $Y' = Y
% \cup Y^+$ such that $x \cgt' y^{\max}$ for all $x \in Y^+$. Accordingly $C^n$ is extended to $C'^n \subset Y'^n$. Let also $xp \sim
% y^{\max}p$ for all $xp, y^{\max}p \in C'^n$. Apparently, we can extend the additive representation onto $C'^n$ by letting $U(x) =
% U(y^{\max})$ for all $x \in Y^{+}$. Hence, strong monotonicity is not, in fact, a necessary condition for the existence of an additive
% representation on a comonotonic cone.

In the Choquet integral representation problem for a heterogeneous product set $X = X_1 \times X_2$, strong monotonicity is actually a
necessary condition because of the following. Assume $ap, bp, cp, dp \in \SE{}$ and $ap \cgt bp, cp \sim dp$. Assume also there exist $cq,dq
\in \NW{}$ such that $cq \cgt dq$. Then, provided the representation exists, we get

\begin{equation*}
  \begin{aligned}
    \alpha_1f_1(a) + \alpha_2f_2(p) & > \alpha_1f_1(b) + \alpha_2f_2(p) \\
    \alpha_1f_1(c) + \alpha_2f_2(p) & = \alpha_1f_1(d) + \alpha_2f_2(p) \\
    \beta_1f_1(c) + \beta_2f_2(q) & > \beta_1f_1(d) + \beta_2f_2(q).
  \end{aligned}
\end{equation*}
The first inequality entails $\alpha_1 \neq 0$. From this and the following equality follows $f_1(c) = f_1(d)$, which contradicts with the
last inequality. Thus $cq \cgt dq$ implies $cp \cgt dp$ but only in the presence of $ap \cgt bp$ in the same ``region'' ($\SE{}$ or $
\NW{}$). This is also the reason behind the name we gave to this condition - ``bi-independence''.

\begin{lemma} Pointwise monotonicity.\par
  \label{lm:4}
  If for all $i,j \in N$ we have $a_ix_j \cgeq a_iy_j$ for all $a_i \in X_i$, then $x \cgeq y$.
\end{lemma}

\begin{proof}
  $x = x_1x_2 \cgeq x_1y_2 \cgeq y_1y_2 = y$.
\end{proof}

Bi-independence, together with the structural assumption also implies some sort of ``strong monotonicity''. 

\begin{lemma}
  \label{lm:29}
  If $X_1$ is essential on $\SE{}(\NW{})$, $a \cgeq_1 b$ iff $ap \cgt bp$ for all $ap,bp \in \NW{}$.
\end{lemma}
\begin{proof}
  Let $X_1$ be essential on $\SE{}$. By the structural assumption, $a \cgeq_1 b$ implies existence of some $q \in X_2$
  such that $aq \cgt bq$. If $aq,bq \in \SE{}$ the result follows by independence. If $aq, bq \in \NW{}$ the result follows by
  bi-independence. If $bq \in \NW{}, aq \in \SE{}$, then by closedness assumption there exists $cq \in \Theta$, and either $bq \cgt cq$, or $cq
  \cgt aq$. The result follows by independence or bi-independence.
\end{proof}

Conceptually, Lemma \ref{lm:29} implies that if a coordinate is essential on some 3C-cone $\NW{z}(\SE{z})$, then it is also essential on $\NW{x}(\SE{x})$
for all $x \in \NW{}(\SE{})$. This allows us to make statements like ``coordinate $i$ is essential on $\NW{}$''.

\section{Building additive value functions on $\NW{}$ and $\SE{}$}
\label{sec:build-addit-repr}

In this section we assume that $\SE{}(\NW{})$ has two essential coordinates.

\subsection{Covering $\SE{}$ and $\NW{}$ with maximal $\SE{z}$ and $\NW{z}$}
\label{sec:covering-se-nw}

In the sequel we could have covered areas $\SE{}$ and $\NW{}$ by sets $\SE{z}(\NW{z})$ for all $z \in \SE{}(\NW{})$, but it is convenient to
introduce the following lemma.

\begin{lemma}
  \label{lm:21}
  For every $x \in \SE{}$ there exists $z \in \Theta$ such that $\SE{x} \subset \SE{z}$. Accordingly, for every $y \in \NW{}$ there exists
  $z \in \Theta$ such that $\NW{y} \subset \NW{z}$.
\end{lemma}
\begin{proof}
  Take $x \in \SE{}$ such that $x \not \in \NW{}$. If there exists $y \in \NW{x}$ such that $y \in \NW{}$, then by closedness assumption there
  must exist either $ay_2 \in \Theta$ and $x \in \SE{ay_2}$ or $x_1p \in \Theta$ and $x \in \SE{x_1p}$. If such $y$ does not exist, $X =
  \Theta$. Other cases are symmetrical.
\end{proof}
It follows from Lemma \ref{lm:21} that $\SE{} = \bigcup _{z \in \Theta} \SE{z}$, while $\NW{} = \bigcup _{z \in \Theta} \NW{z}$. Comparing
this to definitions of $\SE{}$ and $\NW{}$ we are able to define also the following notions: 
\begin{definition}
  We write $x \in \SE{ext}$ and say that $x \in X$ is \emph{extreme in} $\SE{}$ if $x \in \Theta \text{ and }[x_2 \text{ is minimal or } x_1 \text{ is maximal}]$.
  We write $x \in \NW{ext}$ and say that $x \in X$ is \emph{extreme in} $\NW{}$ if $x \in \Theta \text{ and }[x_1 \text{ is minimal or } x_2 \text{ is maximal}]$.
  $x \in X$ is extreme if it is extreme in $\SE{}$ or in $\NW{}$.
\end{definition}
Note that contrary to the homogeneous case $X = Y^n$, extreme points for $\SE{}$ and $\NW{}$ can be asymmetric, i.e. if a point $z$ is
extreme in $\SE{}$ it is not necessarily extreme in $\NW{}$.

\subsection{Representations within $\SE{z}$}
\label{sec:repr-with-ez}

In the following we will build an additive representation on $\SE{}$. The case of $\NW{}$ is symmetric. We proceed by building
representations on sets $\SE{z}$ for all $z \in \Theta \setminus \SE{ext}$ (i.e. for all non-extreme points of $\Theta$). 

\paragraph{Essential coordinates.} For now we assume that both coordinates are essential on $\NW{}$ and $\SE{}$.

\begin{theorem}
  \label{theo:1}
  For any $z \in \Theta \setminus \SE{ext}$ there exists an additive representation of $\cgeq$ on $\SE{z}$:
  \begin{equation*}
    x \cgeq y \Leftrightarrow V^z_1(x_1) + V^z_2(x_2) \geq V^z_1(y_1) + V^z_2(y_2).
  \end{equation*}
\end{theorem}

\begin{proof}
  $\SE{z}$ is a Cartesian product, $\cgeq$ is a weak order on $\SE{z}$, $\cgeq$ satisfies triple cancellation on $\SE{z}$, $\cgeq$ satisfies
  Archimedean axiom on $\SE{z}$, both coordinates are essential. It remains to  show that $\cgeq$ satisfies restricted solvability on
  $\SE{z}$.
  
  Assume that for some $xa, y, xc \in \SE{z}$, we have $xa \cgeq y \cgeq xc$, hence exists $b \in X_2: xb \sim y$. We need to show that $xb
  \in \SE{z}$. If $xb \sim xa$ or $xb \sim xc$, then the result is immediate. Hence, assume $xa \cgt xb \cgt xc$. By definition, $x \in
  \SE{z}$, if $x \cgeq_1 z_1$, and $z_2 \cgeq_2 b$. For $xb$ we need to check only the latter condition. It holds, since $xa \cgt xb \cgt
  xc$, and by weak separability $a \cgeq_2 b$.

  Therefore all conditions for the existence of an additive representation are met \citep{krantz1971foundation}.

\end{proof}

\subsection{Joining representations for different $\SE{z}$ (or $\NW{z}$)}
\label{sec:join-repr-diff}

This section closely follows \mbox{\citep{wakker1991additive-RO}}.

\begin{theorem}
  \label{theo:2}
  There exists an additive interval scale $V^{SE}$ on $\bigcup \SE{z}$, with $z \in \Theta \setminus \SE{ext}$, which represents $\cgeq$ on
  every $\SE{z}$.
\end{theorem}
\begin{proof} 

  Choose the ``reference'' points - pick any $r \in \SE{}$ and any $r^0, r^1 \in \SE{r}$ such that $r^1_1p \cgeq r^0_1p$ for every $p \in
  X_2$. Set $V^{r}_1(r^0_1) = 0, V^{r}_2(r^0_2) = 0 , V^{r}_1(r^1_1) = 1$.  Now, we align representations on the other sets $\SE{z}$ with
  the reference one.  Assume that for some $z \in \Theta$ we have already obtained an additive representation $V^z$ on $\SE{z}$.  Observe
  that $V^z$ and $V^r$ are additive value functions for $\cgeq$ on $\SE{z} \cap \SE{r}$. Morevover $\SE{z} \cap \SE{r} = \SE{q}$, where $q_1
  = r_1$ if $r_1 \cgeq_1 z_1$ and $q_1 = z_1$ if the opposite is true. Similarly, $q_2 = r_2$ if $az_2 \cgeq ar_2$ for
  all $a \in X_1$ and $q_2 = z_2$ in the opposite case. Hence, uniqueness results from \citet{krantz1971foundation} can be applied. In
  particular, this means that on $\SE{z} \cap \SE{r}$ we have $V^r_i = \alpha V^z_i + \beta_i$, so the functions are defined up to a common
  unit and location.

  We choose the unit and location of $V^z_1$ so that $V^z_1(x_1) = V^r_1(x_1)$ for all $ x \in \SE{z} \cap \SE{r}$. Next, we choose the location of $V^z_2$ so
  that it coincides with $V^r_2$ on $\SE{z} \cap \SE{r}$.

  Finally, we show that $V^s_i(x_i) = V^t_i(x_i)$ for any $s,t \in \Theta$ and $x \in \SE{s} \cap \SE{t}$. This immediately follows, since
  $V^s$ and $V^t$ coincide (with $V^r$) on $\SE{s} \cap \SE{t} \cap \SE{r}$. This defines their unit and locations, hence they also coincide
  on $\SE{s} \cap \SE{t}$.  Now define $V^{SE}$ as a function which coincides with $V^{z_i}$ on the respective domains
  $\SE{z_i}$. By the above argument, this function is well-defined.
\end{proof}

\begin{theorem}
  \label{theo:4}
  Representation $V^{SE}$ obtained in Theorem \ref{theo:2} is globally representing on $\SE{} \setminus \SE{ext} = \bigcup _{z \in \Theta \setminus \SE{ext}} \SE{z}$.
\end{theorem}
\begin{proof}

  Let $x \cgeq y$. There can be two cases. First, assume that $x_2 \cgeq_2 y_2$, but $y_1 \cgeq_1 x_1$ (or vice versa). In this case, $x$
  and $y$ belong to the same $\SE{z}$ (e.g. $\SE{x}$) and therefore $V^{SE}$ is a valid representation.

Next, assume that $x_j \cgeq_j y_j$ for all $i,j \in N$. Assume, that $x \in \SE{s}, y \in \SE{t}$. Observe that
  $x_1y_2 \in \SE{s} \cap \SE{t}$ because by the made assumptions, $x_1y_2 \in \SE{x}, x_1y_2 \in \SE{y}$.
  % \begin{equation*}
  %   \begin{aligned}
  %     a_1x_2 & \cgeq a_1y_2 \cgeq a_1t_2 &  \forall a_1 \in X_1 & ~~ (\SE{t}) \\
  %     t_1a_2 & \cgeq x_1a_2 \cgeq y_1a_2 &  \forall a_2 \in X_2 & ~~ (\SE{s}) 
  %       \end{aligned}
  % \end{equation*}
By pointwise monotonicity (Lemma \ref{lm:4}) $x_1x_2 \cgeq x_1y_2 \cgeq y_1y_2$, hence $V_1(x_1) + V_2(x_2) \geq V_1(x_1) +
V_2(y_2) \geq V_1(y_1) + V_2(y_2)$, with first inequalities lying in $\SE{s}$, and second in $\SE{t}$. The reverse implication is also true. 
\end{proof}

% \paragraph*{Remark.}
% \label{sec:remark}
% Since $\SE{} = \bigcup_{z \in \SE{ext}} \SE{z}$, it is easy to see that $\bigcup \SE{z}$ where $z \in \SE{ext}$ with boundary points
% excluded, is equal to $\SE{}$ with exactly boundary extreme points excluded.  
  
% \subsection{Monotonicity}
% \label{sec:monotonicity}

\section{Aligning $V^{SE}$ and $V^{NW}$}
\label{sec:join-togeth-addit}

First we will show that it is not possible for the common domain of $V^{SE}_i$ and $V^{NW}_i$ for some $i$ to contain a single point.

\subsection{Analysis of the common domain of $V^{SE}$ and $V^{NW}$}
\label{sec:technical-lemmas-2}

\begin{lemma} % Special cases of $\SE{}$ and $\NW{}$ - horizontal and vertical division. \par
  \label{lm:10}
  Let $a_0 \cgeq_1 b_0$, and for some $p \in X_2$ we have $a_0p, b_0p \in \Theta$. Define $X_{a_0b_0} = \{x_1: x_1 \in X_1, b_0 \cgeq_1 x_1
  \cgeq_1 a_0\}$.  Then, triple cancellation holds everywhere on $X_{a_0b_0} \times X_2$.
\end{lemma}

\begin{figure}[h!]
  \centering

\begin{tikzpicture}[scale=1.2]

    % invis axes

    \draw[dashed, name path=a0vert] (-3,-3) -- (-3,1);
    \draw[dashed, name path=b0vert] (3,-3) -- (3,1);
    \draw[dashed, name path=a0b0] (-3.5,0) coordinate (p) -- (3.5,0);

    \draw (-3,-3 |- -3.5,0) node[font=\scriptsize, below left] {$b_0$};
    \draw (3,-3 |- -3.5,0) node[font=\scriptsize, below right] {$a_0$};
    \draw (p) node[font=\scriptsize, left] {$p$};

    \fill[black] (-2.5,0.5) coordinate (ax) circle (1.5pt);
    \fill[black] (1.5,0.5) coordinate (cx) circle (1.5pt);

    \fill[black] (-1.5,-0.5) coordinate (by) circle (1.5pt);
    \fill[black] (2.0,-0.5) coordinate (dy) circle (1.5pt);

    \draw (ax) -- (-2.3,0) coordinate (fp) node[font=\scriptsize, above right=1pt] {$f$};
    \draw (fp) -- (-1.5, -0.7) coordinate (ae1);

    \draw (cx) -- ++($(fp) - (ax)$) coordinate (gp);
    \draw (gp) -- ++($(ae1) - (fp)$);

    \fill[black] (-2.5,-1.5) coordinate (aw) circle (1.5pt);
    \fill[black] (1.5,-1.5) coordinate (cw) circle (1.5pt);

    \fill[black] (-1.5,-2.5) coordinate (bz) circle (1.5pt);
    \fill[black] (2.0,-2.5) coordinate (dz) circle (1.5pt);

    \draw[name path=aw2] (aw) -- ++($1.4*(ae1) - 1.4*(fp)$);
    \draw (cw) -- ++($1.4*(ae1) - 1.4*(fp)$);

    \path[name path=fpproj] (fp) -- +(0,-4);
    \draw[dashed, name intersections={of=aw2 and fpproj, by=fq}] (fp) -- (fq);

    \draw ($(-3,0.5) - (1.5pt,0)$) -- ($(-3,0.5) - (-1.5pt,0)$) node[font=\scriptsize, left=3pt] {$x$};
    \draw ($(-3,-0.5) - (1.5pt,0)$) -- ($(-3,-0.5) - (-1.5pt,0)$) node[font=\scriptsize, left=3pt] {$y$};

    \draw ($(-3,-1.5) - (1.5pt,0)$) -- ($(-3,-1.5) - (-1.5pt,0)$) node[font=\scriptsize, left=3pt] {$w$};
    \draw ($(-3,-2.5) - (1.5pt,0)$) -- ($(-3,-2.5) - (-1.5pt,0)$) node[font=\scriptsize, left=3pt] {$z$};

    \path (fq) -- (fq -| -3,0) coordinate (q);
    \draw ($(q) - (1.5pt,0)$) -- ($(q) - (-1.5pt,0)$) node[font=\scriptsize, left=3pt] {$q$};

    \draw ($(-2.5,0) - (0,1.5pt)$) -- ($(-2.5,0) - (0,-1.5pt)$) node[font=\scriptsize, below left=1pt and -1pt] {$a$};
    \draw ($(-1.5,0) - (0,1.5pt)$) -- ($(-1.5,0) - (0,-1.5pt)$) node[font=\scriptsize, below right=1pt and -1pt] {$b$};

    \draw ($(2,0) - (0,1.5pt)$) -- ($(2,0) - (0,-1.5pt)$) node[font=\scriptsize, below right=1pt and -1pt] {$d$};
    \draw ($(1.5,0) - (0,1.5pt)$) -- ($(1.5,0) - (0,-1.5pt)$) node[font=\scriptsize, below left=1pt and -1pt] {$c$};

\end{tikzpicture}

  \caption{Lemma \ref{lm:10}}
  \label{fig:lm10}
\end{figure}

\begin{proof}
  All points in the below proof are from $X_{a_0b_0} \times X_2$.  Let $ax \cleq by, aw \cgeq bz, cx \cgeq dy$. We will show that together
  with the assumptions of the Lemma, this implies $cw \cgeq dz$.
  \\
  \\
  The case when all points belong to $\SE{}$ or $\NW{}$, or two pairs belong to $\SE{}$ and two to $\NW{}$ is covered by \textbf{A4}. Thus,
  assume wlog $x \cgeq_2 p$, so that $ax, cx \in \NW{}$ and the remaining points are in $\SE{}$ (Fig. \ref{fig:lm10}).  Assume also $dp
  \cgeq cp$ and $b \cgeq_1 a$.  Assume also $ax \cgt ap$ (hence by independence also $cx \cgt cp$), $bp \cgt by$ (hence
  also $dp \cgt dy$), otherwise the result immediately follows by \textbf{A4} (e.g. if $ax \sim ap$, we can replace $ax$ by $ap$ and $cx$ by
  $cp$ in the assumptions of the lemma, which brings all points to $\SE{}$).

  \begin{enumerate}
  \item $ax \cleq by$. \\
    $bp \cgt by$, hence $ ax \clt bp$. $ax \cgt ap$, hence $bp \cgt by \cgeq ax \cgt ap$, $bp \cgt ap$,
    therefore, by restricted solvability exists $fp \sim ax$. Also, $fp \cgt ap, bp \cgt fp$.

  \item $cx \cgeq dy$.
    There can be two cases:
    \begin{enumerate}[label={\alph*})]
    \item If $cx \cleq dp$, then $dp \cgeq cx \cgt cp$, hence exists $gp \sim cx$.
    \item $cx \cgt dp$.
    \end{enumerate}
  \item $aw \cgeq bz$. \\
    Solve for $q$: $aw \sim fq$. % Since $bp \cgeq ap$ for all $p$, we have $bz \cgeq az$, hence $aw \cgeq bz \cgeq az$.
    By the results in point 1 and independence we have $fw \cgt aw \cgeq bz \cgt fz$, therefore by restricted solvability exists $q: fq \sim
    aw$.
  \item
    Cases correspond to those in point 2 above: 
    \begin{enumerate}[label={\alph*})]
    \item $fp \sim ax, gp \sim cx, aw \sim fq$, hence by \textbf{A4} $cw \sim gq$ \\
      $fp \cleq by, gp \cgeq dy, fq \cgeq bz$, hence by \textbf{A4} $gq \cgeq dz$ and $cw \cgeq dz$.
    \item $ax \clt bp, cx \cgt dp, aw \cgeq bz$, hence by \textbf{A4} $cw \cgeq dz$. 
    \end{enumerate}

  \end{enumerate}
\end{proof}

From this it follows that it is impossible that for some $i$ the common domain of $V^{SE}_i$ and $V^{NW}_i$ includes a single point. Let
(wlog) $i=1$ and $a \in X_1$ be such a point. Apparently $ap \in \Theta$ for all $a \in X_1$. Then, from Lemma ~\ref{lm:10} it follows that
$\SE{}=\NW{}=X$.

% Let $\sum V^{SE}_i$, $i = 1,2$ be an additive representation of $\cgeq$ on $\SE{}$ (with boundary points of $\SE{ext}$ removed), and $\sum
% V^{NW}_i$ a representation on $\NW{}$.  Assume first that for some $i$ the common domain of $V^{SE}_i$ and $V^{NW}_i$ includes a single
% point. Let (wlog) $i=1$ and $a \in X_1$ be such a point. Apparently $ap \in \Theta$ for all $a \in X_1$. Then, according to Lemma
% ~\ref{lm:hor-spl-add}, there exists an additive representation of $\cgeq$ on $X$.

  % \begin{lemma}
  %   \label{lm:hor-spl-add}
  % Let $a_0q \cgeq b_0q$, for all $q \in X_2$, and exist $ap, bp \in \Theta$. Define $X_{a_0b_0} = \{x_1: x_1 \in X_1, b_0p \cgeq x_1p \cgeq a_0p\}$ for all $p \in X_2$. 
  % Then there exist an additive representation of $\cgeq$ on $X_{a_0b_0} \times X_2$. 
  % \end{lemma}
    
  % \begin{proof}
  %   Follows from Lemma \ref{lm:10}.
  % \end{proof}

\subsection{Aligning representations on \textbf{SE} and \textbf{NW}}
\label{sec:align-addit}

There can be four cases, depending on the number of essential coordinates on $\NW{}$ and $\SE{}$:
\begin{enumerate}
\item Both areas have two essential coordinates;
\item One area has two essential coordinate, another has one essential coordinate;
\item Both areas have one essential coordinates;
\item An area does not have any essential coordinates.
\end{enumerate}

We start with the case where both coordinates are essential on $\NW{}$ and $\SE{}$.

\subsubsection{Both coordinates are essential}
\label{sec:both-coordinates-are}

\begin{lemma}
  \label{lm:7}
  Choose $r^0_1 \in X_1$ and $r^1_1 \in X_1$ from the common domain of $V^{SE}_1$ and $V^{NW}_1$ such that $r^1_1 \cgeq_1 r^0_1$, and set
  $V^{SE}_1(r^0_1) = V^{NW}_1(r^0_1) = 0$ and $V^{SE}_1(r^1_1) = V^{NW}_1(r^1_1) = 1$. Then, $V^{SE}_1(x) = V^{NW}_1(x)$ on all $x$ from
  their common domain.
\end{lemma}
\begin{proof}
  This follows directly from \textbf{A4}. Assume, we want to show that $V^{SE}_1(y_1) = V^{NW}_1(y_1)$ for some $y_1 \in X_1$.  Starting
  from $r^0_1$ build any standard sequence on $X_1$ in $\SE{}$, say $\{\alpha^{(i)}_1 : \alpha^{(i)}_1y^s_1 \sim
  \alpha^{(i+1)}_1y^s_2\}$. Then, all $\alpha^{(i)}_1 y^n_1, \alpha^{(i)}_1 y^n_2$ which are in $\NW{}$ also form a sequence in $\NW{}$: if
  $\alpha^{(i)}_1 y^s_1 \sim \alpha^{(i+1)}_1 y^s_2, \alpha^{(i+1)}_1 y^s_1 \sim \alpha^{(i+2)}_1 y^s_2$ and $\alpha^{(i)}_1 y^n_1 \sim
  \alpha^{(i+1)}_1 y^n_2$, for some $y^n_1, y^n_2 \in X_2$, then by \textbf{A4}, necessarily $\alpha^{(i+1)}_1 y^n_1 \sim \alpha^{(i+2)}_1
  y^n_2$. 

  Now let 
  \begin{equation*}
    \begin{aligned}
      1 = & V^{SE}_1(r^1_1) \approx  n[V^{SE}_2(y^s_2) - V^{SE}_2(y^s_1)] \\
      & V^{SE}_1(y_1) \approx m[V^{SE}_2(y^s_2) - V^{SE}_2(y^s_1)]  \approx \frac{m}{n}.
    \end{aligned}
  \end{equation*}
  Such $n$ and $m$ exist by the Archimedean axiom.
  By the argument above we get 
    \begin{equation*}
    \begin{aligned}
      1 = & V^{NW}_1(r^1_1) \approx  n[V^{NW}_2(y^n_2) - V^{NW}_2(y^n_1)] \\
      & V^{NW}_1(y_1) \approx m[V^{NW}_2(y^n_2) - V^{NW}_2(y^n_1)]  \approx \frac{m}{n}\\
    \end{aligned}
  \end{equation*}

By denserangedness, approximations become exact in the limit, so we obtain $V^{SE}_1 (y_1) = V^{NW}_1 (y_1)$ on
  all $y_1 \in X_1$ from their common domain.
\end{proof}

\begin{lemma}
  \label{lm:8}
  Assume, $V^{SE}$ is an additive representation of $\cgeq$ on $\SE{} \setminus \SE{ext}$, and $V^{NW}$ is a representation on $\NW{}
  \setminus \NW{ext}$, with $V^{SE}_1$ and $V^{NW}_1$ scaled so that they have a common zero and unit (as in Lemma \ref{lm:7}). Then,
  $V^{SE}_2 = \lambda V^{NW}_2$ on the common domain.
\end{lemma}
\begin{proof}
  By Lemma \ref{lm:7}, $V^{SE}_1 = V^{NW}_1$ on the common domain. Assume $V^{SE}_2(r^2_2) = \lambda, V^{NW}_2(r^2_2) = 1$. We will now show
  that $V^{SE}_2 (x_2) = \lambda V^{NW}_2 (x_2)$ for all $x_2 \in X_2$ from the common domain of these functions. Construct a standard
  sequence within $\SE{z}$, this time on $X_2$. By \textbf{A4}, it is also a sequence in $\NW{}$. We obtain

  \begin{equation*}
    \begin{aligned}
      \lambda = & V^{SE}_2(r^2_2) \approx  n[V^{SE}_1(x^s_1) - V^{SE}_1(x^s_1)] \\
      & V^{SE}_2(x_2) \approx m[V^{SE}_1(x^s_1) - V^{SE}_1(x^s_1)]  \approx \frac{\lambda m}{n}\\
    \end{aligned}
  \end{equation*}

  By the argument above we get 
    \begin{equation*}
    \begin{aligned}
      1 = & V^{NW}_2(r^2_2) \approx  n[V^{NW}_2(x^n_1) - V^{NW}_2(x^n_1)] \\
      & V^{NW}_2(x_2) \approx m[V^{NW}_1(x^n_1) - V^{NW}_1(x^n_1)]  \approx \frac{m}{n}\\
    \end{aligned}
  \end{equation*}
  From this in the limit we obtain $V^{SE}_2 (x_2) = \lambda V^{NW}_2 (x_2)$ on all $x_2 \in X_2$ from the common domain of $V^{SE}_2$ and $V^{NW}_2$.
\end{proof}

At this point we can drop superscripts and say that we have representations $V_1 + V_2$ on $\SE{}$ and $V_1 + \lambda V_2$ on $\NW{}$. Fix
two non-extreme points in $\Theta: r^0$ and $r^1$, such that $r^1_1 \cgeq_1 r^0_1$ and $r^1_2 \cgeq_2 r^0_2$. If such points do not exist,
then by Lemma \ref{lm:10} triple cancellation holds everywhere and $\cgeq$ can be represented by an additive function (i.e. $\lambda
=1$). Rescale $V_1$ and $V_2$ so that $V_1(r^0_1) = 0, V_2(r^0_2) = 0, V_1(r^1_1) = 1$. Assume that after rescaling we get
$V_1(r^1_2)=k$. Define $\phi_2(x_2) = V_2(x_2)/k$, i.e. $\phi_2(r^1_2) = 1$. Define $\phi_1(x_1) = V_1(x_1)$. Thus we get representations
$\phi_1 + k\phi_2$ on $\SE{}$ and $\phi_1 + \lambda k \phi_2$ on $\NW{}$.  Finally rescale in the following way: $\frac{1}{1+k}\phi_1 +
\frac{k}{1+k}\phi_2$ on $\SE{}$ and $\frac{1}{1+\lambda k}\phi_1 + \frac{\lambda k}{1 + \lambda k}\phi_2$ on $\NW{}$. We have thus defined
the following representations:
\begin{equation}
  \label{eq:phi_def}
  \begin{aligned}
    \phi^{SE}(x) & = \frac{1}{1+k}\phi_1(x_1) + \frac{k}{1+k}\phi_2(x_2) \\
    \phi^{NW}(x) & = \frac{1}{1+\lambda k}\phi_1(x_1) + \frac{\lambda k}{1 + \lambda k}\phi_2(x_2).
  \end{aligned}
\end{equation}

Note, that it follows that $\phi^{SE}(r^1) = \phi^{NW}(r^1) = 1$.

\subsubsection{One area has a single essential coordinate}
\label{sec:areas-with-single}

Assume $\SE{}$ has two essential coordinates and $\NW{}$ only has $X_1$ essential. After an additive representation $V^{SE}$ has been built
on $\SE{}$, and re-scaled as in (\ref{eq:phi_def}) we have values $\phi_1$ and $\phi_2$ for all points in $\SE{}$, in particular those in
$\Theta$. Let $\phi^{NW}(x) = \phi_1(x_1) + 0\phi_2(x_2)$ (in other words, set $\lambda = 0$ in (\ref{eq:phi_def})) for those $x_i$ where
$\phi_i$ are defined. By structural assumption, bi-independence and additivity $\phi^{NW}$ represents $\cgeq$ on those points for which it
is defined. For example, let $ap, bp \in \NW{}$ be such that $ap \cgt bp$. Since both coordinates are essential on $\SE{}$ by
bi-independence we get also $aq \cgt bq$ for all $q \in X_2$ such that $aq, bq \in \SE{}$. Additivity implies $\phi_1(a) > \phi_1(b)$. For
the remaining $x_1 \in X_1$, i.e. for $x_1 \in X_1$ such that there are no points in $\Theta$ first coordinate of which is $x_1$, build a
simple ordinal representation. Lemma \ref{lm:34} shows that values for all $x_2 \in X_2$ have already been defined at this point. Other cases are similar.
% \linebreak \linebreak
% \hl{!!!!!!!!!!!!!!!!!!!!!!!!!!!MOVE!!!!!!!!!!!!!!!!!!!!!!!!!!!!!}\\
% By Theorem \ref{theo:3} all $x \in \Theta$ get the same value in both $\NW{}$ and $\SE{}$ - $\phi^{NW}(x) = \phi^{SE}(x)$ for all $x \in
% \Theta$.
% \hl{!!!!!!!!!!!!!!!!!!!!!!!!!!!MOVE!!!!!!!!!!!!!!!!!!!!!!!!!!!!!}\\

% Now assume $\SE{}$ only has one essential coordinate and $\NW{}$ also has only one essential coordinate. The procedure is effectively the
% same - first build an ordinal representation on $\SE{}$ and then, using points from $\Theta$, set values for the remaining points in $X_1$
% and $X_2$. 

\subsubsection{Both areas have a single essential coordinate}
\label{sec:both-areas-have}

An interesting result is that \textbf{A3} is sufficient for characterization of cases where both $\SE{}$ and $\NW{}$ have one essential
coordinate. There are two cases in total:
\begin{enumerate}
\item $X_1$ is essential on $\NW{}$, $X_2$ is essential on $\SE{}$;
\item $X_2$ is essential on $\NW{}$, $X_1$ is essential on $\SE{}$.
\end{enumerate}

We will need the following lemma.

\begin{lemma}
  \label{lm:25}
  Let $X_1$ be essential on $\NW{}$ and $X_2$ be essential on $\SE{}$ or $X_2$ be essential on $\NW{}$ and $X_1$ be essential on
  $\SE{}$. Then, either for all $x \in \SE{}$ exists $z \in \Theta$ such that $z \sim x$, or for all $y \in \NW{}$ exists $z \in \Theta$
  such that $z \sim y$.
\end{lemma}
\begin{proof}
  We only consider one case, others being symmetrical. Assume $X_1$ is essential on $\SE{}$ and $X_2$ is essential on $\NW{}$. Assume also
  there exists $x \in \SE{}$ such that $x \cgt z$ for all $z \in \Theta$, in particular some maximal $z^{\max}$. We will show that this
  implies that there does not exist $y \in \NW{}$ such that $y \cgt z$ or $z \cgt y$ for all $z \in \Theta$.

  Assume such $y$ exists. Take $x_1y_2$. By \textbf{A3} it belongs either to $\SE{}$ or $\NW{}$. If it belongs to $\NW{}$, by closedness
  assumption exists $x_1t \in \Theta$. We get $x_1t \sim x \cgt z^{\max}$ - a contradiction. If $x_1y_2 \in \SE{}$, exists $ay_2 \in
  \Theta$. We have $y \sim ay_2$ which contradicts both $y \cgt z$ and $z \cgt y$ for all $z \in \Theta$.

  Finally, we need to show that $\Theta$ does not have gaps. Assume there exists $y \in \NW{}$ and $z^1,z^2 \in \Theta$ such that $z^1 \cgt
  y \cgt z^2$ but there is no $z \in \Theta$ such that $z \sim y$. Since only $X_2$ is essential on $\NW{}$, we get $z^1_1y_2 \in
  \SE{}$. By closedness assumption exists $x_1 \in X_1$ such that $x_1y_2 \in \Theta$. Since only $X_2$ is essential we conclude $x_1y_2 \sim y$,
  which is a contradiction. Therefore, for every $y \in \NW{}$ there exists $z \in \Theta$ such that $y \sim z$.
\end{proof}

\paragraph{Defining value functions.}
\label{sec:defin-value-funct-1}

Lemma \ref{lm:25} guarantees that for all points in $\SE{}$ or all points in $\NW{}$ exists an equivalent point in $\Theta$. Assume for example
that $\Theta$ is such that for all $x \in \SE{}$ exists $z \in \Theta$. Assume also that $X_1$ is essential on $\SE{}$ and $X_2$ is
essential on $\NW{}$. Now define value functions $\phi_1:X_1 \rightarrow \mathbb{R}$ and $\phi_2:X_2 \rightarrow \mathbb{R}$ as follows.
Choose $\phi_1$ to be any real-valued function such that $\phi_1(x_1) > \phi_1(y_2)$ iff $x_1 \cgeq_1 y_1$. Now for all
$z$ from $\Theta$ set $\phi_2(z_2) = \phi_1(z_1)$. Finally, extend $\phi_2$ to the whole $X_2$ by choosing any function such that
$\phi_2(x_2) > \phi_2(y_2)$ iff $x_2 \cgeq_2 y_2$. Lemma \ref{lm:34} guarantees that the functions have been defined for
all $x_1 \in X_1$ and all $x_2 \in X_2$.

\subsubsection{Areas without essential coordinates}
\label{sec:areas-without-ess}

\begin{lemma}
  \label{lm:14}
  If \textbf{A1 - A9} and the structural assumption hold, there can not be $\NW{z}(\SE{z})$ with no essential coordinates.
\end{lemma}
\begin{proof}
  Assume for some $z \in \Theta$ the set $\NW{z}$ has no essential coordinates. By bi-independence and the structural assumption
  it follows that there are no essential coordinates on any $\NW{z}$. This implies (by \textbf{A7}) that both coordinates are essential on
  $\SE{}$. Take $ap,bp \in \NW{z}$. Apparently, $ap \sim bp$. By structural assumption there must exist $q \in X_2$ such that $aq \cgt
  bq$. It can't be that $aq, bq \in \NW{}$, hence $aq, bq \in \SE{}$.

  By closedness assumption there exist $w, z \in X_2$ such that $aw, bz \in \Theta$. Also, since no coordinate is essential in $\NW{}$ we have
  $aw \sim bz$.  Since $aq \cgt bq$ it must be $bz \cgt bw$, since otherwise by strict monotonicity (Lemma \ref{lm:29}) it can't be that $aw
  \sim bz$.

  By independence we have $aw \cgt bw$. By definition of $\NW{}$ and $\SE{}$ we have $aw \in \SE{}$ (since by weak separability $a \cgeq_1
  b$) and $aw \in \NW{}$ (since by weak separability $z \cgeq_2 w$). Hence, by independence it must
  be $az \cgt aw$ (since $az \in \SE{}$) and $az \sim aw$ (since $az \in \NW{}$). We have arrived at a contradiction.
\end{proof}
\section{Properties of the intersection of $\SE{}$ and $\NW{}$}
\label{sec:0-set}

\begin{lemma}
  \label{lm:12}
  For any non-extreme $x \in X$ we have:
  \begin{equation*}
    x \in \Theta \Rightarrow \phi_1(x_1) = \phi_2(x_2),
  \end{equation*}
unless $\cgeq$ can be represented by an additive function (i.e $\lambda = 1$ in (\ref{eq:phi_def})).
\end{lemma}

For the case when both $\NW{}$ and $\SE{}$ have a single essential coordinate the result holds by definition of $\phi_i$, so for the
remainder of this section we assume that $\SE{}$ or $\NW{}$ has two essential coordinates.

\begin{figure}[h!]

\begin{tikzpicture}[scale=1.5]

    % invis axes
    \path [name path=axis] (-1,5) |- (5.5,0) ;

    \coordinate (r0) at (1,1);
    \draw[dashed, name path=r0xaxis] (r0) -- (5.5,1) coordinate (r0x5);
    \draw[dashed, name path=r0yaxis] (r0) -- (1,4.7) coordinate (r0y5);
    % Draw two intersecting lines
    \draw[name path=diag] (0,0) coordinate (a_1) -- (4,4) coordinate (a_2);
    \path (1.5,0) coordinate (b_1) -- (1,1) coordinate (b_2) -- (0,1.3) coordinate (b_3);
    % Calculate the intersection of the lines a_1 -- a_2 and b_1 -- b_2
    % and store the coordinate in c.
    \coordinate (c) at (intersection of a_1--a_2 and b_1--b_2);

    \fill[black] (r0) circle (1.5pt) node [below left=4pt and -3pt] {$r^0$};

    \draw [dashed] (1.5,0) -- ($(r0)!(b_1)!(r0x5)$);
    \fill [black] ($(r0)!(b_1)!(r0x5)$) circle (1pt);% node [font=\scriptsize, above right] {$\alpha^{(1)}_1$};
    \draw let \p1 = (b_1) in (\x1,1.5pt) -- (\x1,-1.5pt) node[font=\scriptsize, below] {$\alpha^{(1)}_1$};

    % \coordinate (r1) at (3.8,3.8);
    % \fill[black] (r1) circle (1.5pt) node [above=4pt] {$r^1$};

    \coordinate (dir1) at ($4*(b_2)-4*(b_1)$);
    \coordinate (dir2) at ($4*(b_3)-4*(b_2)$);

    \foreach  \t in {1.5, 2, ..., 4}
    {
    \coordinate (t) at (\t,0);
    \path[name path=line1_\t] (t) -- +(dir1);
    \draw[name intersections={of=diag and line1_\t,by={Int1}}] (t) -- (Int1);
    \path[name path global/.expanded=line2_\t] (Int1) -- +(dir2);
    \path[name intersections={of=axis and line2_\t,by={Int2}}] (Int1) -- (Int2);
  }
  
    \path[name path=line1] (t) -- +(dir1);
 
    \path[name intersections={of=line2_2 and r0yaxis, by=beta1}, name path={beta1_y}] (beta1) -- ++(-5,0);
    \path[name intersections={of=beta1_y and line2_1.5, by={b1y_1}}];
    \fill[black] (beta1) circle (1pt);
    \draw[dashed] (beta1) -- (b1y_1);
    \draw ($(b1y_1) - (1.5pt,0)$) -- ($(b1y_1) - (-1.5pt,0)$) node[font=\scriptsize, left=3pt] {$\beta^{(1)}_1$};

%    \fill [name intersections={of=r0yaxis and line2, by={b_5}}] (b_5) circle (1pt);

    % vis axes
    \draw [<->,thick, name path=axis] let \p1 = (b1y_1) in (\x1,5) node (yaxis) [left] {$X_2$}
        |- (5.5,0) node (xaxis) [below] {$X_1$};

    \foreach  \t in {1.5, 2, ..., 4}
    {
    \coordinate (t) at (\t,0);
    \path[name path global/.expanded=line1_\t] (t) -- +(dir1);
    \path[name intersections={of=diag and line1_\t,by={Int1}}] (t) -- (Int1);
    \path[name path global/.expanded=line2_\t] (Int1) -- +(dir2);
    \draw[name intersections={of=axis and line2_\t,by={Int2}}] (Int1) -- (Int2);
  }
  \fill[black] (Int1) circle (1.5pt) node [below left=4pt and -3pt] {$x$};
% x projection onto X1
  \path[name intersections={of=line1_4 and r0xaxis, by=px}];
  \fill[black] (px) circle (1pt);
  \draw let \p1 = (px) in (\x1,1.5pt) -- (\x1,-1.5pt) node[font=\scriptsize, below] {$\pi_1(x)$};
% x projection onto X2
  \path[name intersections={of=line2_4 and r0yaxis, by=py}];
  \fill[black] (py) circle (1pt);
  \path (py) -| (py -| b1y_1) coordinate (py_ax);
  \draw ($(py_ax) - (1.5pt,0)$) -- ($(py_ax) - (-1.5pt,0)$) node[font=\scriptsize, left=3pt] {$\pi_2(x)$};

    \coordinate (t) at (4.5,0);
    \path[name path=line1] (t) -- +(dir1);
    \draw[name intersections={of=diag and line1,by={Int1}},name path=sl6] (t) -- (Int1);
    \path [name intersections={of=r0xaxis and sl6, by={a_6}}] (a_6) -- (Int1);

    \path[name path=line2] (Int1) -- +(dir2);
    \draw[name intersections={of=axis and line2,by={Int2}}] (Int1) -- (Int2);
        
    \coordinate (t) at (5,0);
    \path[name path=line1_5] (t) -- +(dir1);
    \draw[name intersections={of=diag and line1_5,by={Int1}}, name path=sl7] (t) -- (Int1);
    \fill[black] (Int1) circle (1.5pt) node [below left=4pt and -3pt] {$r^1$};
    \path[name path=line2_5] (Int1) -- +(dir2);
    \draw[name intersections={of=axis and line2_5,by={Int2}}] (Int1) -- (Int2);
% r1 projection on X1
    \path[name intersections={of=line1_5 and r0xaxis, by=pr1}];
    \fill[black] (pr1) circle (1pt);
    \draw let \p1 = (pr1) in (\x1,1.5pt) -- (\x1,-1.5pt) node[font=\scriptsize, below] {$\pi_1(r^1)$};
% r1 projection onto X2
  \path[name intersections={of=line2_5 and r0yaxis, by=pr1_2}];
  \fill[black] (pr1_2) circle (1pt);
  \path (pr1_2) -| (pr1_2 -| b1y_1) coordinate (pr1_ax);
  \draw ($(pr1_ax) - (1.5pt,0)$) -- ($(pr1_ax) - (-1.5pt,0)$) node[font=\scriptsize, left=3pt] {$\pi_2(r^1)$};

%    \draw[dashed] (r0) -- ($(0,0)!(r0)!(0,5)$);
    % \draw let \p1 = (yaxis) in ($3 + 1.5pt$,0) -- ($3 - 1.5pt$,\y1) node[font=\scriptsize, left] {$y_1$};
    % \draw let \p1 = (0,0) in (1.5pt,\y1) -- (-1.5pt,\y1) node[font=\scriptsize, left] {$y_0$};
    
    \coordinate (t) at (5.5,0);
    \path[name path=line1] (t) -- +(dir1);
    \path[name intersections={of=diag and line1,by={Int1}}, name path=sl8] (t) -- (Int1);

\end{tikzpicture}

  \caption{Lemma \ref{lm:12}}
  \label{fig:theo3_easy}
\end{figure}

\begin{proof}

  % First assume that $x$ is non-boundary.
  We start with a case where both coordinates are essential on $\SE{}$ and $\NW{}$. Assume also $x \cgeq r^0$ (without loss of generality,
  other cases are symmetrical and can be proved by the same technique). We are going to show that $x \in \Theta \Rightarrow \phi_1(x_1) =
  \phi_2(x_2)$ or $\lambda = 1$.

Assume for now that we can find the following solutions:
 
\begin{itemize}
\item Solve for $\pi_1(r^1): \pi_1(r^1)r^0_2 \sim r^1$.
\item Solve for $\pi_2(r^1): r^0_1\pi_2(r^1) \sim r^1$.
\item Solve for $\pi_1(x):  \pi_1(x)r^0_2 \sim x$.
\item Solve for $\pi_2(x): r^0_1\pi_2(x) \sim x$.
\end{itemize}

Pick $\alpha^{(1)}_1 \in X_1$ such that $r^1_1 \cgeq_1 \alpha^{(1)}_1$, it exists by denserangedness. Solve for $\beta^{(1)}_2 :
\alpha^{(1)}_1 r^0_2 \sim r^0_1 \beta^{(1)}_2$, exists by restricted solvability.

Now build an increasing standard sequence $\alpha^{(i)}_1: \alpha^{(0)}_1 = r^0_1$ on $X_1$ which lies in $\SE{}$ and an increasing standard
sequence $\beta^{(i)}_2: \beta^{(0)}_2 = r^0_2$ on $X_2$ which lies in $\NW{}$ (see Fig. \ref{fig:theo3_easy}).  Since $\pi_1(r^1)r^0_2 \sim
r^0_1\pi_2(r^1)$, by \textbf{A5} (Lemma \ref{lm:A5}) we have for some $m$:
\begin{equation*}
  \alpha^{(m-1)}_1r^0_2 \cgeq \pi_1(r^1)r^0_2 \cgeq \alpha^{(m)}_1 r^0_2  \iff r^0_1\beta^{(m-1)}_2 \cgeq r^0_1\pi_2(r^1) \cgeq r^0_1 \beta^{(m)}_2.
\end{equation*}

From this (and since $\phi_1(r^0_1) = \phi_2(r^0_2) = 0$) it follows that $\phi_1(\pi_1(r^1)) \approx m\phi_1(\alpha^{(1)}_1), \phi_2(\pi_2(r^1)) \approx m\phi_2(\beta^{(1)}_2)$ and,
since $\phi^{SE}(\pi_1(r^1)r^0_2) = \phi^{SE}(r^1) = \phi^{NW}(r^1) = \phi^{NW}(r^0_1\pi_2(r^1))$, we obtain: 
\begin{equation}
\label{eq:1}
  \frac{1}{1+k}m\phi_1(\alpha^{(1)}_1)  = \frac{\lambda k}{1 + \lambda k}m\phi_2(\beta^{(1)}_2).
\end{equation}

Similarly, $\pi_1(x)r^0_2 \sim r^0_1\pi_2(x)$, so by \textbf{A5} (Lemma \ref{lm:A5}) we have
\begin{equation*}
  \alpha^{(n-1)}_1r^0_2 \cgeq \pi_1(x)r^0_2 \cgeq \alpha^{(n)}_1 r^0_2  \iff r^0_1\beta^{(n-1)}_2 \cgeq r^0_1\pi_2(x) \cgeq r^0_1 \beta^{(n)}_2.
\end{equation*}

From this follows that $\phi_1(\pi_1(x)) \approx n\phi_1(\alpha^{(1)}_1), \phi_2(\pi_2(x)) \approx n \phi_2(\beta^{(1)}_2)$ and by (\ref{eq:1}) it
follows  that $\phi^{SE}(\pi_1(x)r^0_2) = \phi^{NW}(r^0_1\pi_2(x))$. Hence
\begin{equation*}
  \frac{1}{1+k}\phi_1(x_1) + \frac{k}{1+k}\phi_2(x_2) = \frac{1}{1+\lambda k}\phi_1(x_1) + \frac{\lambda k}{1 + \lambda k}\phi_2(x_2),
\end{equation*}
and so $\phi_1(x_1) = \phi_2(x_2)$ or $\lambda = 1$ (i.e. the structure is additive). 

We need to revisit the case where solutions mentioned in the beginning do not exist. Consider Figure \ref{fig:theo3}. Assume this time that
there does not exist $\pi_1(r^1)$ such that $r^1 \sim \pi_1(r^1)r^0_2$. If we choose the step in the standard sequence $\alpha^{(i)}_1$
small enough so that there exist $\alpha^{(k+1)}_1, \alpha^{(k+2)}_1$ such that $\alpha^{(k+1)}_1r^0_2 \cgeq r^1_0r^0_2$ and
$\alpha^{(k+2)}_2r^0_2 \cgeq r^1_0r^0_2$ (which we can do by non-maximality of $r^1$ and denserangedness of $\cgeq$), then we can ``switch''
from the standard sequence $\alpha^{(i)}_1$ on $X_1$ to the standard sequence $\gamma^{(i)}_2$ on $X_2$ keeping the same increment in value
between subsequent members of the sequence.  Indeed, $\gamma^{(k)}_2: r^1_1\gamma^{(k)}_2 \sim \alpha^{(k)}_1r^0_2$ and $\gamma^{(k+1)}_2:
r^1_1\gamma^{(k+1)}_2 \sim \alpha^{(k+1)}_1r^0_2$ exist by monotonicity and restricted solvability, so does $x_1 \colon x_1\gamma^{(k)}_2 \sim
r^1_1\gamma^{(k+1)}_2 \sim \alpha^{(k+1)}_1r^0_2$ % and $\gamma^{(k+2)}_2: r^1_1 \gamma^{(k+2)}_2 \sim x_1 \gamma^{(k+1)}_2$
, and by \textbf{A5} (Lemma \ref{lm:A5}) we
get $[r^1_1\gamma^{(k)}_2 \sim r^0_1\beta^{(k)} _2, r^1_1\gamma^{(k+1)}_2 \sim r^0_1\beta^{(k+1)}_2] \Rightarrow \gamma^{(i)}_2 \sim
\beta^{(i)}_2 $ for all $i$. Note that $\gamma^{(i)}_2r^1_2$ are in $\SE{}$ for all $i$ such that $r^1 \cgeq \gamma^{(i)}_2$ since $r^0$ and
$r^1$ are in $\SE{}$. By monotonicity and restricted solvability there exists $i$ such that $x_1r^1_2 \cgeq x_1\gamma^{(i+1)}_2 \cgeq r^1
\cgeq x_1\gamma^{(i)}_2$.  Finally, the increment in value is the same between members of $\alpha_i$ and $\gamma^{(i)}_2$ since $\alpha_k
\sim \gamma^{(k)}_2$ and $\alpha_{k+1} \sim \gamma^{(k+1)}_2$. The result then follows as above.

\begin{figure}[h!]

\begin{tikzpicture}[scale=1.5]

    % Draw axes
    \draw [<->,thick, name path=axis] (0,5) node (yaxis) [left] {$X_2$}
        |- (5.5,0) node (xaxis) [below] {$X_1$};

    \coordinate (r0) at (1,1);
    \draw[dashed, name path=r0xaxis] (r0) -- (5.5,1) coordinate (r0x5);
    \draw[dashed, name path=r0yaxis] (r0) -- (1,4.7) coordinate (r0y5);
    % Draw two intersecting lines
    \draw[name path=diag] (0,0) coordinate (a_1) -- (4,4) coordinate (a_2);
    \draw (1.5,0) coordinate (b_1) -- (1,1) coordinate (b_2) -- (0,1.3) coordinate (b_3);
    % Calculate the intersection of the lines a_1 -- a_2 and b_1 -- b_2
    % and store the coordinate in c.
    \coordinate (c) at (intersection of a_1--a_2 and b_1--b_2);

    \fill[black] (r0) circle (1.5pt) node [below left=4pt and -3pt] {$r^0$};

    \draw [dashed] (1.5,0) -- ($(r0)!(b_1)!(r0x5)$);
    \fill [black] ($(r0)!(b_1)!(r0x5)$) circle (1pt);% node [font=\scriptsize, above right] {$\alpha^{(1)}_1$};
    \draw let \p1 = ($(r0)!(b_1)!(r0x5)$) in (\x1,1.5pt) -- (\x1,-1.5pt) node[font=\scriptsize, below] {$\alpha^{(1)}_1$};
    \coordinate (r1) at (3.8,3.8);
    \fill[black] (r1) circle (1.5pt) node [above=4pt] {$r^1$};

    \coordinate (dir1) at ($4*(b_2)-4*(b_1)$);
    \coordinate (dir2) at ($4*(b_3)-4*(b_2)$);

    \foreach  \t in {1.5, 2, ..., 4}
    {
    \coordinate (t) at (\t,0);
    \path[name path=line1] (t) -- +(dir1);
    \draw[name intersections={of=diag and line1,by={Int1}}] (t) -- (Int1);
    \path[name path=line2] (Int1) -- +(dir2);
    \draw[name intersections={of=axis and line2,by={Int2}}] (Int1) -- (Int2);
  }
  
    \path[name path=line1] (t) -- +(dir1);
    \path[name path=p6] (r1) -| (xaxis -| r1);
    \draw[name intersections={of=line1 and p6,by={a2_5}}] (r1) -- (a2_5);

    \fill[black] (a2_5) circle (1pt); % node [font=\scriptsize, above right] {$\gamma^{(k)}_2$};
    \draw let \p1 = (a2_5) in (1.5pt,\y1) -- (-1.5pt,\y1) node[font=\scriptsize, left] {$\gamma^{(k)}_2$};
    \fill [name intersections={of=r0xaxis and line1, by={a_5}}] (a_5) circle (1pt); % node [above right] {\scriptsize $\alpha^{(k)}_1$};
    \draw let \p1 = (a_5) in (\x1,1.5pt) -- (\x1,-1.5pt) node[font=\scriptsize, below] {$\alpha^{(k)}_1$};

    \fill [name intersections={of=r0yaxis and line2, by={b_5}}] (b_5) circle (1pt);% node [font=\scriptsize, above right] {$\beta_i$};
    \draw let \p1 = (b_5) in (1.5pt,\y1) -- (-1.5pt,\y1) node[font=\scriptsize, left] {$\beta^{(k)}_2$};

    \coordinate (t) at (4.5,0);
    \path[name path=line1] (t) -- +(dir1);
    \draw[name intersections={of=diag and line1,by={Int1}},name path=sl6] (t) -- (Int1);
    \path [name intersections={of=r0xaxis and sl6, by={a_6}}] (a_6) -- (Int1);
    \fill[black] (a_6) circle (1pt);% node [font=\scriptsize, below right] {$\alpha^{(k+1)}_1$};
    \draw let \p1 = (a_6) in (\x1,1.5pt) -- (\x1,-1.5pt) node[font=\scriptsize, below left=0pt and -12pt] {$\alpha^{(k+1)}_1$};
    \path[name path=line2] (Int1) -- +(dir2);
    \draw[name intersections={of=axis and line2,by={Int2}}] (Int1) -- (Int2);
    
    \coordinate (a2_6) at (intersection of r1--a2_5 and t--Int1);
%    \fill[black] (a2_6) circle (1pt);% node[font=\scriptsize, above left] {$\gamma^{(k+1)}_2$};
    \draw let \p1 = (a2_6) in (1.5pt,\y1) -- (-1.5pt,\y1) node[font=\scriptsize, left] {$\gamma^{(k+1)}_2$};
    \path [name path = ss2_unit] (a2_5)--+(1,0);
    \draw[dashed,name intersections={of=ss2_unit and sl6,by={ss2}}] (a2_5) -- (ss2);
    \fill[black] (ss2) circle (1pt);
    \draw[dashed, name path = ss2_low] let \p1 = (ss2) in (\x1,4) -- ($(0,0)!(ss2)!(6,0)$);% node[font=\scriptsize, below=6pt] {$x_1$};    
    \draw let \p1 = ($(0,0)!(ss2)!(6,0)$) in (\x1,1.5pt) -- (\x1,-1.5pt) node[font=\scriptsize, below=4.5pt] {$x_1$};
    \draw[dashed] (a2_5) -- ($(0,0)!(a2_5)!(6,0)$);    
    
    \fill [name intersections={of=r0yaxis and line2, by={b_6}}] (b_6) circle (1pt);% node [font=\scriptsize, above right] {$\beta_{i+1}$};
    \draw let \p1 = (b_6) in (1.5pt,\y1) -- (-1.5pt,\y1) node[font=\scriptsize, left] {$\beta^{(k+1)}_2$};

    \coordinate (t) at (5,0);
    \path[name path=line1] (t) -- +(dir1);
    \path[name intersections={of=diag and line1,by={Int1}}, name path=sl7] (t) -- (Int1);
    \path[name path=line2] (Int1) -- +(dir2);
    \draw[name intersections={of=axis and line2,by={Int2}}] (Int1) -- (Int2);

    \coordinate (a2_7) at (intersection of r1--a2_5 and t--Int1);
    \path[name path=tint] (t)--(Int1);
    \path [name intersections={of=ss2_low and tint, by=g2_7}];
    \fill[black] (g2_7) circle (1pt);% node [font=\scriptsize, above right] {$\gamma^{(k+2)}_2$};
  
  \draw let \p1 = (a2_7) in (1.5pt,\y1) -- (-1.5pt,\y1) node[font=\scriptsize, left] {$\gamma^{(k+2)}_2$};
    \draw (a2_7) -- (Int1);
    \draw[name intersections={of=sl7 and r0xaxis, by={a7}}] (a2_7) -- (a7);
    \fill[black] (a7) circle (1pt);% node [font=\scriptsize, above right] {$\alpha^{(k+2)}_1$};
    \draw let \p1 = (a7) in (\x1,1.5pt) -- (\x1,-1.5pt) node[font=\scriptsize, below right] {$\alpha^{(k+2)}_1$};
    \fill [name intersections={of=r0yaxis and line2, by={b_7}}] (b_7) circle (1pt);% node [font=\scriptsize, above right] {$\beta_{i+2}$};
    \draw let \p1 = (b_7) in (1.5pt,\y1) -- (-1.5pt,\y1) node[font=\scriptsize, left] {$\beta^{(k+2)}_2$};

    \draw[dashed] (r0) -- ($(0,0)!(r0)!(0,5)$);% node[left] {$y_1$};
    \draw let \p1 = ($(0,0)!(r0)!(0,5)$) in (1.5pt,\y1) -- (-1.5pt,\y1) node[font=\scriptsize, left] {$r^0_2$};
%    \draw (0,0) node[left] {$y_0$};
    \draw let \p1 = (0,0) in (1.5pt,\y1) -- (-1.5pt,\y1) node[font=\scriptsize, left] {$y_0$};
    %\draw [dashed] (r1) -- ($(r0)!(r1)!(r0x5)$);
    
    \coordinate (t) at (5.5,0);
    \path[name path=line1] (t) -- +(dir1);
    \path[name path=r1_g2_dir] (r1) -- ++($-1*(dir1)$);
    \draw[name intersections={of=r1_g2_dir and ss2_low, by=r1_g2}] (r1) -- (r1_g2);
    \fill[black] (r1_g2) circle (1pt);

    \path[name intersections={of=diag and line1,by={Int1}}, name path=sl8] (t) -- (Int1);
    \coordinate (a8) at (intersection of r1--a2_5 and t--Int1);
    \fill[black] let \p1 = (ss2), \p2 = (a8) in (\x1,\y2) circle (1pt);

    \draw let \p1 = (a8) in (1.5pt,\y1+2pt) -- (-1.5pt,\y1+2pt) node[font=\scriptsize, below left=-4pt and 0pt] {$\gamma^{(k+3)}_2$};
    \path[name intersections={of=ss2_low and line1,by=Int3}];
    \fill[black] (Int3) circle (1pt);
    \draw[name path=sl81] (Int3) -- (Int1);
    \path[name path=line2] (Int1) -- +(dir2);
    \draw[name intersections={of=axis and line2,by={Int2}}] (Int1) -- (Int2);

\end{tikzpicture}

  \caption{Lemma \ref{lm:12} - changing direction}
  \label{fig:theo3}
\end{figure}
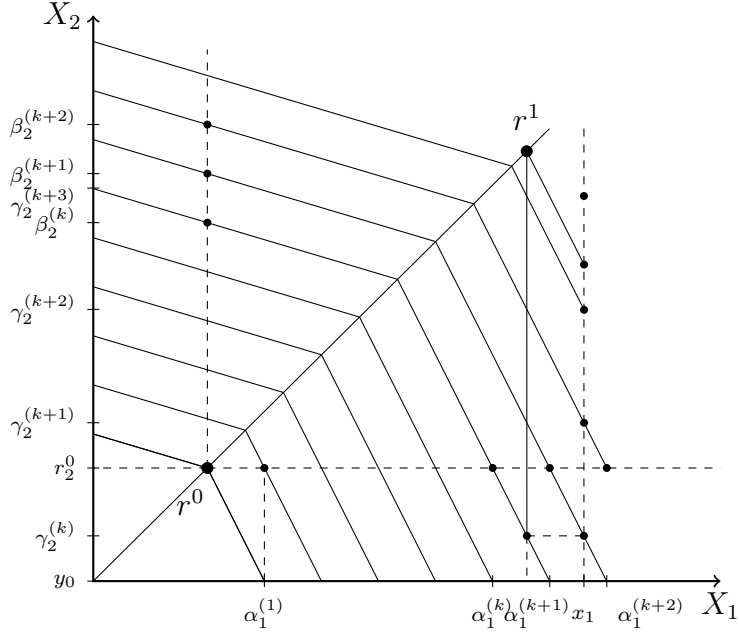

Finally, we look at the case where only one coordinate is essential on either $\NW{}$ or $\SE{}$. First assume that $X_2$ is essential on
$\NW{}$. We defined $\phi^{NW}(x) = 0\phi_1(x_1) + \phi_2(x_2)$. Definition implies $\phi^{NW}(r^0)=0, \phi^{NW}(r^1)=1$. Build a standard
sequence $\{\alpha_1^{(i)}\}$ on $X_1$ from $r^0$ to $r^1$ (in case there exists a solution for $r^1 \sim \pi_1(r^1)r^0_2$, otherwise use the
approach detailed in the previous paragraph), setting $\alpha^{(0)}_1 = r^0_1$. Take $\alpha^{(1)}_1r^0_2$ and $\alpha^{(2)}_1r^0_2$. By restricted
solvability there must exist $\beta^{(1)}_2$ and $\beta^{(2)}_2$, such that $\alpha^{(1)}_1r^0_2 \sim r^0_1\beta^{(1)}_2$ and $\alpha^{(2)}_1r^0_2 \sim
r^0_1\beta^{(2)}_2$.  By closedness assumption for $\beta^{(1)}_2, \beta^{(2)}_2$ there must exist $x_1,x_2$ such that $x_1\beta^{(1)}_2 \in \Theta$,
$x_2\beta^{(2)}_2 \in \Theta $. Also, since only $X_2$ is essential, we get $x_1\beta^{(1)}_2 \sim r^0_1\beta^{(1)}_2$, $x_2\beta^{(2)}_2 \sim
r^0_1\beta^{(2)}_2$. By weak monotonicity and definition of $\SE{}$, $\alpha^{(2)}_1\beta^{(1)}_2 \cgeq \alpha^{(2)}_1r^0_2 \sim x_2\beta^{(2)}_2 \cgeq
x_2\beta^{(1)}_2$, hence by restricted solvability exists $z_1: x_2\beta^{(2)}_2 \sim z_1\beta^{(1)}_2$. By \textbf{A5} then $x_2\beta^{(1)}_2 \sim
z_1r^0_2$. By additivity $x_2\beta^{(2)}_2 \sim z_1\beta^{(1)}_2$ and $x_2\beta^{(1)}_2 \sim z_1r^0_2$ entail $\phi_2(b_2) - \phi_2(b_1) = \phi_2(b_1) -
\phi_2(r^0_2)$. From this the result follows as in the proof above. If now $X_1$ is essential on $\NW{}$ repeat the proof as above this time
starting the sequence from $r^1$ ``towards'' $r^0$.

\end{proof}

\begin{lemma}
  \label{cor:crit_order}
The following statements hold or $\cgeq$ has an additive representation: 
\begin{enumerate}
\item If $ap \in \Theta$ then for no $b \in X_1$ holds $bp \in \Theta$ and also for no $q \in X_2$ holds $aq \in \Theta$.
% \item $\SE{ext} = \NW{ext} = \Theta$
\item $x \in \SE{} \Rightarrow \phi_1(x_1) \geq \phi_2(x_2)$, $y \in \NW{} \Rightarrow \phi_2(y_2) \geq \phi_1(y_1)$. 
\end{enumerate}  
\end{lemma}

\begin{proof}
  \begin{enumerate}
  \item Assume $ap, bp \in \Theta$. Assume first, that $\SE{}$ or $\NW{}$ has two essential coordinates. If $ap \sim bp$, then $aq \sim bq$
    for all $q \in X_2$, which violates the structural assumption. Let $ap \cgt bp$. By denserangedness there exists also $cp$ such that $ap
    \cgt cp \cgt bp$. If $p$ is minimal, then by definition of $\SE{}$, $bp \in \SE{}$ if for no $c \in X_1$ such that $c \cgeq_1 b$ holds
    $cp \in \NW{}$. This is obviously violated by $ap$. If $p$ is maximal, then by definition of $\NW{}$, $ap \in \NW{}$ if for no $c \in
    X_1$ such that $a \cgeq_1 c$ holds $cp \in \SE{}$, which is violated by $bp$. Hence, $p$ is not maximal or minimal. If $a$ is maximal,
    then $cp, bp$ are non-extreme, hence by Lemma \ref{lm:12}, $\phi_1(c) = \phi_2(p) = \phi_1(b)$, hence $bq \sim cq$ for all $q \in
    X_2$. If $b$ is minimal, $ap, bp$ are non-extreme, hence $\phi_1(a) = \phi_2(p) = \phi_1(c)$, hence $aq \sim cq$ for all $q \in X_2$. In
    both cases structural assumption is violated.
  % \item Take any $ap \in \Theta$. Assume it is not extreme in $\SE{}$, in particular exists $bp \in \SE{}$ such that $\SE{ap} \subset
  %   \SE{bp}$. This entails $aq \cgeq bq$ for all $q \in X_2$, hence $bp \in \NW{} \Rightarrow bp \in \Theta$. This contradicts the previous statement.
  \item Pick any $bq \in \SE{}$. By Lemma \ref{lm:21} there exists $ap \in \Theta$ such that $bq \in \SE{ap}$, hence $b \cgeq_1 a$ and $p
    \cgeq_2 q$. By Lemma \ref{lm:12} $\phi_1(a) = \phi_2(p)$. We also have $\phi_1(b) \geq \phi_1(a), \phi_2(p) \geq \phi_2(q)$. The result
    follows. $\NW{}$ case is symmetric.
  \end{enumerate}
\end{proof}

\section{Extending value functions to extreme points}
\label{sec:extreme-points}

Value functions for the case when both $\SE{}$ and $\NW{}$ have a single essential coordinate were fully defined in Section
\ref{sec:both-areas-have}. Thus in what follows we will consider cases where $\SE{}$ or $\NW{}$ have two essential coordinates. 

As indicated in \citep{wakker1991additive-RO}, value functions might be driven to infinite values at the maximal/minimal points of
rank-ordered subsets, nevertheless not implying existence of infinite standard sequences residing entirely within comonotonic cones. Put it
another way, it might be not possible to ``reach'' a maximal/minimal point with a sequence lying entirely in $\NW{}$ or $\SE{}$. Yet another
way to say it is that for some maximal/minimal point $z$, the set $\NW{z}(\SE{z})$ contains no standard sequences (see also
\citep{wakker1991additive-RO} Remark 24).

The cornerstone of this section is Lemma \ref{lm:12}. It plays the same role as proportionality of value functions plays in
\citep{wakker1991additive-RO}, effectively guaranteeing that both value functions $\phi_1$ and $\phi_2$ are limited if maximal/minimal
elements exist.

\begin{lemma}
  \label{lm:19}
  Assume that $\SE{}$ has two essential coordinates. The following statements hold:
  \begin{itemize}
  \item If there exist a maximal $M_1 \in X_1$, $\phi_1$ is bounded from above.
  \item If there exist a minimal $m_2 \in X_2$, $\phi_2$ is bounded from below.
  \end{itemize}
  Assume that $\NW{}$ has two essential coordinates. The following statements hold:
  \begin{itemize}
  \item If there exist a minimal $M_1 \in X_1$, $\phi_1$ is bounded from below.
  \item If there exist a maximal $m_2 \in X_2$, $\phi_2$ is bounded from above.
  \end{itemize}
\end{lemma}
\begin{proof}
  We shall only prove the first one. First, notice that there must exist $p \in X_2$ such that $M_1p \in \SE{}$. Take $x_1 \in X_1$ and
  $v_2, w_2 \in X_2$ such that $v_2 \cgeq_2 w_2$, and $x_1v_2 \in \SE{}$. If such points cannot be found, $X$ has an additive
  representation (all $x \in \NW{}$), and the result follows. So we assume such points exist. By definition of $\SE{x_1v_2}$ it follows that
  $M_1w_2, x_1v_2, M_1v_2 \in \SE{}$. Hence, we can evoke the argument from \cite{wakker1991additive-RO} Lemma 20.

  If $M_1w_2 \cleq x_1v_2$ then we have an upper bound: $V_1(M_1) \leq V_1(x_1) + V_2(v_2) - V_2(w_2)$.  If $M_1w_2 \cgt x_1v_2$ then by
  monotonicity $M_1v_2 \cgeq M_1w_2 \cgt x_1v_2$ and hence exists $z_1 \in X_1$ such that $M_1w_2 \sim z_1v_2$, hence $z_1v_2 \cgeq \beta
  w_2$ for all $\beta \in X_1$.
\end{proof}
\begin{lemma}
  \label{lm:23}
  If $x_1x_2 \in \Theta$ and $x_1x_2$ is extreme, then 
  \begin{equation*}
    \lim _{z \in \Theta, z_2 \rightarrow x_2} \phi_2(z_2) = \lim _{z \in \Theta, z_1 \rightarrow x_1} \phi_1(z_1).    
  \end{equation*}
\end{lemma}
\begin{proof}
  For the case when $\SE{}$ or $\NW{}$ have two essential variables the result follows from Lemma \ref{lm:12}, otherwise it is by
  definition of $\phi_i$ (see Section \ref{sec:both-areas-have}).
\end{proof}

\subsection{Extending value functions to extreme elements of $\Theta$}
\label{sec:extend-bound-points}

\paragraph{Extreme elements of $\Theta$ are the only representatives of maximal/minimal equivalence classes of $\SE{}$($\NW{}$).}
\label{sec:sdf}

\begin{lemma}
  \label{lm:35}
  Let $X_1$ be essential on $\SE{}$. If there exists $z \in \Theta$ such that $z_2$ is minimal, then $x \cgt z$ for all $x \in \SE{}$. If
  $X_2$ is essential on $\SE{}$ and there exists $z \in \Theta$ such that $z_1$ is maximal, then $z \cgt x$ for all $x \in
  \SE{}$. Similarly, if $X_1$ is essential on $\NW{}$ and there exists $z \in \Theta$ such that $z_2$ is maximal, then $z \cgt x$ for all $x
  \in \NW{}$. If $X_2$ is essential on $\NW{}$ and there exists $z \in \Theta$ such that $z_1$ is minimal, then $x \cgt z$ for all $x \in
  \NW{}$.
\end{lemma}

\begin{proof}
  We provide the proof just for one of the cases. Let $\NW{}$ have two essential variables. Assume $z_2$ is maximal. Since $z \in \Theta$,
  for all $x \in \NW{}$ holds $z_1 \cgeq_1 x_1$ and by maximality $z_2 \cgeq_2 x_2$. Hence, by Lemma \ref{lm:29}, $z \cgt x$ for all $x \in
  \NW{}$. The case with the minimal $z_1$ is symmetric.
\end{proof}

\paragraph{Uniqueness of definition of $\phi_i$ at the  extreme elements of $\Theta$. }
\label{sec:uniq-defin-phi_i}

\begin{lemma}
  \label{lm:36}
  If both coordinates are essential on $\SE{}$ and $\NW{}$ the values of $\phi_i$ for extreme $x \in \Theta$ are uniquely defined. Moreover,
  $\phi_1(x_1) = \phi_2(x_2)$.
\end{lemma}
\begin{proof}
  Assume, for example $x_1x_2 \in \Theta$ and $x_1$ is minimal. Then any $z_1x_2$ such that $z_1 \cgeq_1 x_1$, belongs to $\SE{}$, and any
  equivalence relation within $\SE{}$ involving $z_1x_2$ uniquely defines $\phi_2(x_2)$ (see (\ref{eq:phi_def})). Similarly, any $x_1z_2$
  such that $z_2 \cgeq_2 x_2$, belongs to $\NW{}$, and any equivalence relation within $\NW{}$ involving $x_1z_2$ uniquely defines
  $\phi_1(x_1)$. By Lemma \ref{lm:23} these values are equal.
\end{proof}

\begin{lemma}
  \label{lm:37}
  If both coordinates are essential on $\SE{}$ but only one on $\NW{}$ (or vice versa) the values of $\phi_i$ for extreme $x \in \Theta$
  can be set as follows:
  \begin{itemize}
  \item If $x_1x_2 \in \Theta$, $\NW{}$ has two essential coordinates and $x_2$ is maximal, then $\phi_2(x_2)$ is uniquely defined, $\phi_1(x_1)$ can be set to
    any value greater or equal to $\phi_2(x_2)$.
  \item If $x_1x_2 \in \Theta$, $\NW{}$ has two essential coordinates and $x_1$ is minimal, then $\phi_1(x_1)$ is uniquely defined, $\phi_2(x_2)$ can be set to
    any value less or equal to $\phi_1(x_1)$.
  \item If $x_1x_2 \in \Theta$, $\SE{}$ has two essential coordinates and $x_2$ is minimal, then $\phi_2(x_2)$ is uniquely defined, $\phi_1(x_1)$ can be set to
    any value less or equal to $\phi_2(x_2)$.
  \item If $x_1x_2 \in \Theta$, $\SE{}$ has two essential coordinates and $x_1$ is maximal, then $\phi_1(x_1)$ is uniquely defined, $\phi_2(x_2)$ can be set to
    any value greater or equal to $\phi_1(x_1)$.
  \end{itemize}
\end{lemma}
\begin{proof}
  Consider the first case. $\phi_2(x_2)$ is defined uniquely as in the proof of Lemma \ref{lm:36}. However, this is not possible for
  $\phi_1(x_1)$. This is because $x$ is the only point in $\NW{}$ having $x_1$ as the first coordinate, and, by Lemma \ref{lm:35} there is
  no equivalence relation within $\NW{}$ which involves $x$. If $X_1$ is essential on $\SE{}$ then all points from the equivalence class
  which includes $x$ also have $x_1$ as their first coordinate, which does not allow to elicit $\phi_1(x_1)$. If only $X_2$ is essential on
  $\SE{}$, then the representations of equivalences involving $x_1$ do not include $\phi_1(x_1)$.
\end{proof}

In the case where only one coordinate is essential on both $\SE{}$ and $\NW{}$ no special treatment is required for the extreme elements of
$\Theta$.

\begin{lemma}
  \label{lm:6}
  If $x \in \SE{ext}$ then for any $y \in NW{}$ such that $x \sim y$, we have:
  \begin{equation*}
    \phi^{NW}(x) = \phi^{NW}(y).
  \end{equation*}
  If further, $x_1$ is maximal, then $\phi^{SE}(x) > \phi^{SE}(y)$ for all $y \in \SE{}$. If $x_2$ is minimal, then $\phi^{SE}(x) <
  \phi^{SE}(y)$ for all $y \in \SE{}$. \\
  \\
  If $x \in \NW{ext}$ then for any $y \in SE{}$ such that $x \sim y$, we have:
  \begin{equation*}
    \phi^{SE}(x) = \phi^{SE}(y).
  \end{equation*}
  If further, $x_1$ is minimal, then $\phi^{NW}(x) < \phi^{NW}(y)$ for all $y \in \NW{}$. If $x_2$ is maximal, then $\phi^{NW}(x) >
  \phi^{NW}(y)$ for all $y \in \NW{}$.
\end{lemma}

\begin{proof}
  In case when $\SE{}$ and $\NW{}$ have the same number of essential variables, value functions at the extreme points are defined uniquely
  (Lemma \ref{lm:36} for case when both variables are essential and by definition otherwise) and the second parts of each statement follow
  immediately by Lemma \ref{lm:35}. If only $\SE{}$ or $\NW{}$ have two essential variables, the result follows by Lemma \ref{lm:37}, Lemma
  \ref{lm:35}, and definition of $\phi$.
\end{proof}

\begin{lemma}
  \label{lm:24}
  For any $x \in X$ we have:
  \begin{equation*}
    \phi_1(x_1) = \phi_2(x_2) \Rightarrow x \in \Theta.
  \end{equation*}
\end{lemma}

\begin{proof}
  Assume $\phi_1(x_1) = \phi_2(x_2)$ and $x \not \in \NW{}$. By Lemma \ref{lm:21} exists $z \in \Theta$ such that $x \in \SE{z}$. By
  structural assumption we have $\phi_2(z_2) \geq \phi_2(x_2) = \phi_1(x_1) \geq \phi_1(x_1)$, with at least one inequality being strict
  (otherwise $x = z$).

  If $z$ is non-extreme then by Lemma \ref{lm:12} we have $\phi_1(z_1) = \phi_2(z_2)$ - a contradiction. If $z$ is extreme, the only cases
  when $\phi_2(z_2) > \phi_1(z_1)$ can hold is when either $z_2$ is minimal or $z_1$ is maximal. But it is easy to see that in this case the
  only points for which it is not possible to find a non-extreme $z$, are the extreme points themselves.  
\end{proof}

\begin{lemma}
  \label{lm:9}
  For all $x \in X$ such that $\phi_1(x_1) \geq \phi_2(x_2)$ we have $x \in \SE{}$. If $x \in X$ is such that $\phi_2(x_2) \geq
  \phi_1(x_1)$ then $x \in \NW{}$.
\end{lemma}

\begin{proof}
  For non-extreme points this follows from Lemma \ref{cor:crit_order} and Lemma \ref{lm:24}. Assume $\phi_1(x_1) \geq \phi_2(x_2)$. If
  $\phi_1(x_1) = \phi_2(x_2)$ then by Lemma \ref{lm:24} $x \in \Theta$, so we are done. Therefore, assume $\phi_1(x_1) > \phi_2(x_2)$. If $x
  \in \NW{}$, then by Lemma \ref{cor:crit_order} it must be $\phi_2(x_2) \geq \phi_1(x_1)$, a contradiction. Therefore, $x \in \SE{}$. For
  extreme points the result follows from Lemma \ref{lm:37}.
\end{proof}

Finally, we can formulate:

\begin{theorem} The following statements hold:
  \label{theo:3}
  \begin{itemize}
  \item If both $\NW{}$ and $\SE{}$ have two essential variables, then for all $x \in X$:
    \begin{equation*}
      x \in \Theta \iff \phi_1(x_1) = \phi_2(x_2),
    \end{equation*}
    unless $\cgeq$ can be represented by an additive function (i.e $\lambda = 1$ in (\ref{eq:phi_def})).
  \item If only $\NW{}$ or only $\SE{}$ have two essential variables, the for all non-extreme $x \in X$:
    \begin{equation*}
      x \in \Theta \Rightarrow \phi_1(x_1) = \phi_2(x_2),
    \end{equation*}
    while at extreme $x \in X$, $\phi_1(x_1)$ and $\phi_2(x_2)$ are related as in Lemma \ref{lm:37}. Finally, for all $x \in X$:
    \begin{equation*}
      \phi_1(x_1) = \phi_2(x_2) \Rightarrow x \in \Theta,
    \end{equation*}
  \item If both $\NW{}$ and $\SE{}$ have only one essential variable, then for all $x \in X$:
    \begin{equation*}
      x \in \Theta \iff \phi_1(x_1) = \phi_2(x_2).
    \end{equation*}
  \end{itemize}
\end{theorem}

\begin{proof}

Follows from Lemmas \ref{lm:12}, \ref{lm:24}, \ref{lm:36}, \ref{lm:37}.
  
\end{proof}

% \begin{lemma}
%   \label{lm:18}
%   Assume $\NW{}$ has two essential coordinates, $\SE{}$ has one and let $z \in \Theta$ be boundary. Then, if $z_2$ is maximal, $\phi^{SE}(z)
%   > \phi^{NW}(x)$ for all $x \in \NW{} \setminus z$. If $z_1$ is minimal, then $\phi^{NW}(x) > \phi^{SE}(z)$ for all $x \in \NW{} \setminus z$. Similarly, if
%   $\SE{}$ has two essential coordinates and $\NW{}$, then, if $z_1$ is maximal, $\phi^{NW}(z) > \phi^{SE}(x)$ for all $x \in \SE{} \setminus z$. If
%   $z_2$ is minimal, then $\phi^{SE}(x) > \phi^{NW}(z)$ for all $x \in \SE{} \setminus z$.
% \end{lemma}
% \begin{proof}
%   Lemma \ref{lm:37} details what values functions $\phi_i$ can take at the boundary points of $\Theta$.  Consider the case when $\NW{}$ has
%   two essential variables and $\SE{}$ has one - $X_1$. Let $z \in \Theta$ and $z_2$ be maximal. According to Lemma \ref{lm:37} $\phi_1(z_1)$ can be
%   set to any value greater or equal than some limit value $\phi^*$. If it is set to the limit values, $\phi^{NW}$ and $\phi^{SE}$ are equal
%   by Lemma \ref{lm:equal_theta}: $\phi^{NW} = \frac{1}{1+\lambda k}\phi^* + \frac{\lambda k}{1 + \lambda k}\phi^* = \phi* + 0\phi^* =
%   \phi^{SE}$. It also holds by Lemma \ref{lm:35}, that $\phi^* > \phi^{NW}(x)$ for all $x \in \SE{}$. From this, if we set $\phi_1(z_1) >
%   \phi^*$, we immediately get $\phi^{SE}(z) > \phi^{NW}(x)$ for all $x \in \NW{} \setminus z$. The other cases are proved  identically. 
% \end{proof}

% \paragraph{Redefining $\NW{}$ and $\SE{}$}
% \label{sec:redefining-nw-se}

% In the case when 

\section{Constructing a global representation on $X$}
\label{sec:vnw-vse-represent}

\begin{lemma}
  \label{lm:11}
  Assume $z^1 \cgt z^2$ for some $z^1, z^2 \in \Theta$. There exists $z$ such that $z^1 \cgt z \cgt z^2$. 
\end{lemma}
\begin{proof}
  By order density and closedness assumption.
\end{proof}

For all $x \in X$ let $\phi^x(x)$ be equal to $\phi^{SE}(x)$ if $x \in \SE{}$ or $\phi^{NW}(x)$ if $x  \in \NW{}$. For points in $\Theta$
values of two latter functions coincide, so $\phi^x(x)$ is well-defined.   

\begin{lemma}
  \label{lm:20}
  Let $\phi^x(x) > \phi^y(y)$. Then, $x \cgt y$. 
\end{lemma}

\begin{proof}
  If $x$ and $y$ belong to $\SE{}$ or $\NW{}$ the conclusion is immediate, so we only need to look at the remaining case. Assume $x \in
  \SE{}$, $y \in \NW{}$. 

  First we will show that it can't hold that $x \cgt z, y \cgt z$ or $z \cgt x, z \cgt y$ for all $z \in \Theta$. Assume $x \cgt z, y \cgt
  z$ for all $z \in \Theta$. Let $x_1 \cgeq_1 y_1$. If $x_1y_2 \in \SE{}$, then exists $z_1y_2 \in \Theta$ such that $x_1y_2 \cgeq z_1y_2
  \cgeq y_1y_2$, a contradiction. If $x_1y_2 \in \NW{}$, then exists $x_1z_2 \in \Theta$, such that $x_1y_2 \cgeq x_1z_2 \cgeq x_1x_2$,
  again a contradiction. Other cases are symmetrical. 

  Hence, assume there exists $z^1 \in \Theta$, such that $x \cgeq z^1, y \cgeq z^1$ and $z^2 \in \Theta$ such that $z^2 \cgeq x$ or $z^2
  \cgeq y$. The only non-trivial case is $z^2 \cgt x \cgt z^1, z^2 \cgt y \cgt z^1 $ (in other cases one of the points $z^1$ or $z^2$
  immediately leads to the conclusion). We have 
  \begin{equation*}
    \phi(z^2) > \phi(x) > \phi(y) > \phi(z^1),
  \end{equation*}
  hence also $\phi_1(z^2_1) = 0.5 \phi(z^2) > 0.5 \phi(x) > 0.5 \phi(y) > 0.5 \phi(z^1) = \phi_1(z^1_1)$. By denserangedness of $\phi_1$
  (see \citep{wakker1991additive-RO} equation (16)), we can find a point $c_1$ such that $0.5 \phi(x) > \phi(c_1) > 0.5 \phi(y)$. We have
  $c_1z^1_2 \in \SE{}, c_1z^2_2 \in \NW{}$, hence there exists $c_2$ such that $c_1c_2 \in \Theta$. Since $c_1c_2$ is not extreme, we have
  $\phi(c_1c_2) = 2\phi_1(c_1)$, and hence $\phi(x) > \phi(c) > \phi(y)$. The first inequality is in $\SE{}$, while the second is in
  $\NW{}$, hence we conclude that $x \cgt y$. 
\end{proof}

\begin{lemma}
  \label{lm:16}
  Let $x \cgt y$. Then, $\phi(x) > \phi(y)$. 
\end{lemma}
\begin{proof}
  By Lemma \ref{lm:20} we have $x \cgt y \Rightarrow \phi(x) \geq \phi(y)$. Hence, we need to show that $\phi(x) \neq \phi(y)$. 
  Assume, $x \in \SE{}, y \in \NW{}$. If $\SE{}$ or $\NW{}$ have only one essential coordinate, then by Lemma \ref{lm:34} exists $z \in
  \Theta$, equivalent either to $x$ or $y$, from which the conclusion is immediate. Hence, assume both areas have two essential
  coordinates. 

  $x_1$ and $x_2$ can't be both minimal, because otherwise $x \cgt y$ cannot hold by pointwise monotonicity, so assume $x_1$ is not
  minimal. We will find a point $z$ in $\SE{}$ such that $x \cgt z \cgeq y$. Take some $z_1$ such that $x_1 \cgeq_1 z_1$ and $z_1x_2 \in
  \SE{}$ (it can be found by closedness and order density). If $z_1x_2 \cgeq y$, we have $\phi(x) > \phi(z_1x_2) \geq \phi(y)$, otherwise by
  restricted solvability we can find $w_1$ such that $w_1x_2 \sim y$, $w_1 \in \SE{}$, and hence $\phi(x) > \phi(w_1x_2) = \phi(y)$
  (equality follows from Lemma \ref{lm:20}). The case when $x_2$ is not maximal is identical.
\end{proof}

\begin{theorem}
  \label{th:5}
  For any $x, y \in X$ we have 
  \begin{equation*}
    x \cgeq y \iff \phi(x) \geq \phi(y).
  \end{equation*}
\end{theorem}
\begin{proof}
  Immediate by Lemmas \ref{lm:20} and \ref{lm:16}.
\end{proof}

\section{The representation is a Choquet integral}
\label{sec:repr-choq-integr}

Representations $\phi^{SE}$ and $\phi^{NW}$ uniquely define a capacity $\nu$. For the case when $\SE{}$ or $\NW{}$ has two essential
coordinates, set (using (\ref{eq:phi_def})): 
\begin{equation*}
  \begin{aligned}
    \nu(\{1\}) & = \frac{1}{1+k} & \text{ (from } \phi^{SE} \text{)}\\
    \nu(\{2\}) & = \frac{\lambda}{1+\lambda k} & \text{ (from } \phi^{NW} \text{)}\\
    \nu(\{1,2\})  & = 1.
  \end{aligned}
\end{equation*}
Thus, we obtain
\begin{equation*}
  \begin{aligned}
    C(\nu,\phi(x)) = \phi^{SE}(x) & = \frac{1}{1+k}\phi_1(x_1) + \frac{k}{1+k}\phi_2(x_2), & \qquad \text{ for all }x \in \SE{}, \\
    C(\nu,\phi(x)) = \phi^{NW}(x) & = \frac{1}{1+\lambda k}\phi_1(x_1) + \frac{\lambda k}{1 + \lambda k}\phi_2(x_2), & \qquad \text{ for all }x \in \NW{}.
  \end{aligned}
\end{equation*}
Assume now, that $\SE{}$ and $\NW{}$ has only one essential coordinate. If $X_1$ is essential on $\SE{}$ set $\nu({1}) = 1$, otherwise zero.
If $X_2$ is essential on $\NW{}$ set $\nu({2}) = 1$, otherwise zero. As above, set $\nu(\{1,2\}) = 1$
We obtain:
\begin{equation*}
  \begin{aligned}
    C(\nu,\phi(x)) & = \phi_1(x_1) , & \qquad \text{ if }X_1 \text{ is essential on the area containing } x, \\
    C(\nu,\phi(x)) & = \phi_2(x_2) , & \qquad \text{ if }X_2 \text{ is essential on the area containing } x,
  \end{aligned}
\end{equation*}
in particular $C(\nu,\phi(x)) = \max(\phi_1(x_1),\phi_2(x_2))$ if $X_1$ is essential on $\SE{}$, $X_2$ is essential on $\NW{}$,
$C(\nu,\phi(x)) = \min(\phi_1(x_1),\phi_2(x_2))$ if $X_2$ is essential on $\SE{}$, $X_1$ is essential on $\NW{}$.

% \begin{lemma}
%   \label{lm:26}
%   Set $\nu(\emptyset) = \nu(1) = \nu(2) = 0, \nu({1,2})=1$. Let $\phi_1, \phi_2$ be as defined in Section \ref{sec:both-areas-have}. Then
%   $C_{\nu}(\phi_1(x_1), \phi_2(x_2)) \geq C_{\nu}(\phi_1(y_1), \phi_2(y_2))$ iff $x \cgeq y$.
% \end{lemma}
% \begin{proof}
%   By construction of $\phi_1$ and $\phi_2$, it holds $\phi_1(x_1) = \phi_2(x_2)$ if $x \in \Theta$. By definition of $\SE{}$ and $\NW{}$ and
%   definition of $\phi_1, \phi_2$ it follows $\phi_1(y_1) \geq \phi_2(y_2)$ for all $y \in \SE{}$, $\phi_2(y_2) \geq \phi_1(y_1)$ for all $y
%   \in \NW{}$.

%   If both $x$ and $y$ are in $\SE{}$ (or $\NW{}$) the result is immediate, as the integral is equal to $\phi_1(z_1)$($\phi_2(z_2)$) for all
%   $z$ in $\SE{}$. Assume then $x \in \SE{}, y \in \NW{}$. By Lemma \ref{lm:25} exists $x_1p \in \Theta$, and by assumptions of the lemma $x
%   \sim x_1p$. From this the result follows.
% \end{proof}

\section{Uniqueness}
\label{sec:uniqueness}

Uniqueness properties are similar to those obtained in the homogeneous case $X = Y^n$, but are modified to accommodate for the heterogeneous
structure of the set $X$ in this paper.

\begin{lemma}
  \label{lm:uniq}
  Representation (\ref{eq:repr}) has the following uniqueness properties:
  \begin{enumerate}
  \item If $\NW{}=\SE{}=X$ then for any functions $g_1:X_1 \rightarrow \mathbb{R}, g_2:X_2 \rightarrow \mathbb{R}$
    such that (\ref{eq:repr}) holds with $f_i$ substituted by $g_i$, we have $f_i (x_i) = \alpha g_i (x_i) + \beta_i$.
  \item If both coordinates are essential on $\NW{}$ and $\SE{}$, then for any functions $g_1:X_1 \rightarrow \mathbb{R}, g_2:X_2
    \rightarrow \mathbb{R}$ such that (\ref{eq:repr}) holds with $f_i$ substituted by $g_i$, we have $f_i (x_i) = \alpha g_i (x_i) +
    \beta$.
  \item If both coordinates are essential on $\NW{}$, but only one coordinate is essential on $\SE{}$, then for any functions $g_1:X_1 \rightarrow \mathbb{R}, g_2:X_2
    \rightarrow \mathbb{R}$ such that (\ref{eq:repr}) holds with $f_i$ substituted by $g_i$, we have : 
    \begin{itemize}
    \item $f_i (x_i) = \alpha g_i (x_i) + \beta$, for all $x$ such that $f_1(x_1) < \max f_2(x_2)$ and $f_2(x_2) > \min f_1(x_1)$;
    \item $f_i(x_i) = \phi_i(g_i(x_i))$ where $\phi_i$ is an increasing function, otherwise.
    \end{itemize}
  \item If both coordinates are essential on $\SE{}$, but only one coordinate is essential on $\NW{}$, then for any functions $g_1:X_1
    \rightarrow \mathbb{R}, g_2:X_2 \rightarrow \mathbb{R}$ such that (\ref{eq:repr}) holds with $f_i$ substituted by $g_i$, we have :
    \begin{itemize}
    \item $f_i (x_i) = \alpha g_i (x_i) + \beta$, for all $x$ such that $f_2(x_2) < \max f_1(x_1)$ and $f_1(x_1) > \min f_2(x_2)$;
    \item $f_i(x_i) = \psi_i(g_i(x_i))$ where $\psi_i$ is an increasing function, otherwise.
    \end{itemize}
  \item If one coordinate is essential on $\NW{}$ and another one on $\SE{}$, then for any functions $g_1:X_1 \rightarrow \mathbb{R},
    g_2:X_2 \rightarrow \mathbb{R}$ such that (\ref{eq:repr}) holds with $f_i$ substituted by $g_i$, we have : $f_i(x_i) =
    \psi_i(g_i(x_i))$ where $\psi_i$ are increasing functions such that $f_1(x_1) = f_2(x_2) \iff g_1(x_1) = g_2(x_2)$.
  \end{enumerate}
\end{lemma}

\begin{proof}
  \begin{enumerate}
  \item Direct by uniqueness properties of additive representations.
  \item Direct by uniqueness properties of additive representations, Lemma \ref{lm:36}, Theorem \ref{th:5}.
  \item Assume $\NW{}$ has two essential coordinates. If there exists an element in $\Theta$ such that $x_1$ is minimal or $x_2$ is maximal,
    then, by Lemma \ref{lm:19}, there exist respectively a minimal $\phi_1(x_1)$ and maximal $\phi_2(x_2)$. Points $x$ such that $f_1(x_1) <
    \max f_2(x_2)$ and $f_2(x_2) > \min f_1(x_1)$ are precisely the elements for which there exist $a \in X_1$ or $p \in X_2$ such that
    either $ax_2 \in \NW{}$ or $x_1p \in \NW{}$. From this follows that uniqueness of $\phi_i$ for these points is defined by the uniqueness
    properties of $\NW{}$ and definition of $\phi_i$ on $\SE{}$, i.e. $f_i(x_i) = \alpha g_i(x_i) + \beta$. For the remaining points
    (including extreme elements of $\Theta$), uniqueness is derived from the uniqueness of ordinal representations and Lemma \ref{lm:34}.
  \item The proof is identical to the one in the previous point.
  \item Uniqueness properties are derived from the uniqueness of ordinal representations, Lemma \ref{lm:34} and definition of $\phi_i$
    (Section ~\ref{sec:both-areas-have}). 
  \end{enumerate}
\end{proof}

Uniqueness part of Theorem \ref{theo:repr} directly follows from Lemma \ref{lm:uniq}.

\section{Necessity of axioms}
\label{sec:neccessity-axioms}

\paragraph{A4.}
Let $ap \cleq bq, ar \cgeq bs, cp \cgeq dq$ and assume $cr \clt ds$. Let also $ap, bq, cp, dq \in \NW{}$, $ar,bs,cr,ds \in \SE{}$ and $X_1$
to be symmetric in $\NW{}$ (the other cases are symmetric). We obtain: 

\begin{equation*}
  \begin{aligned}
    \alpha_1 f_1(a) + \alpha_2 f_2(p) & \leq \alpha_1 f_1(b) + \alpha_2 f_2(q) \\
    \alpha_1 f_1(c) + \alpha_2 f_2(p) & \geq \alpha_1 f_1(d) + \alpha_2 f_2(q) \\
    \beta_1 f_1(a) + \beta_2 f_2(r) & \geq \beta_1 f_1(b) + \beta_2 f_2(s) \\
    \beta_1 f_1(c) + \beta_2 f_2(r) & < \beta_1 f_1(d) + \beta_2 f_2(s)
  \end{aligned}
\end{equation*}
From the first two inequalities and essentiality of $X_1$ ($\alpha_1 \neq 0$) follows $f_1(a) + f_1(d) \leq f_1(b) + f_1(c)$. Last two inequalities imply
$f_1(a) + f_1(d) > f_1(b) + f_1(c)$, a contradiction.

We also give the ``necessity'' proof of the condition in Lemma \ref{lm:A5}, since comparing it with the necessity proof of \textbf{A5}
allows to elicit some interesting implications of essentiality. 

\paragraph{Lemma \ref{lm:A5}}
Let $\{g^{(i)}_1 : g^{(i)}_1y_0 \sim g^{(i+1)}_1y_1, g^{(i)}_1 \in X_1, i \in N \}$ and $\{h^{(i)}_2 : x_0h^{(i)}_2 \sim x_1h^{(i+1)}_2,
h^{(i)}_2 \in X_x, i \in N \}$ be two standard sequences, the first entirely contained in $\NW{}$ and the second in $\SE{}$. Assume also,
that there exist $z_1, z_2 \in X$, $p,q \in X_2, a,b \in X_1$ such that $g^{(i)}_1p, g^{(i)}_1q \in \NW{}$, and $ah^{(i)}_2, bh^{(i)}_2 \in
\SE{}$ for all $i$, and $g^{(i)}_1p \sim bh^{(i)}_2$ and $g^{(i+1)}_1p \sim bh^{(i+1)}_2$. Finally, assume, $g^{(i+2)}_1p \cgt
bh^{(i+2)}_2$.  Other cases are symmetric. 

\begin{equation*}
  \begin{aligned}
    \alpha_1 f_1(g^{(i)}_1) + \alpha_2 f_2(y_0) & = \alpha_1 f_1(g^{(i+1)}_1) + \alpha_2 f_2(y_1) \\
    \alpha_1 f_1(g^{(i)}_1) + \alpha_2 f_2(y_0) & = \alpha_1 f_1(g^{(i+1)}_1) + \alpha_2 f_2(y_1) \\
    \\
    \beta_1 f_1(x_0)  + \beta_2 f_2(h^{(i)}_1) & = \beta_1 f_1(x_1)  + \beta_2 f_2(h^{(i+1)}_1) \\
    \beta_1 f_1(x_0)  + \beta_2 f_2(h^{(i+1)}_1) & = \beta_1 f_1(x_1)  + \beta_2 f_2f_2(h^{(i+2)}_1) \\
    \\
    \alpha_1 f_1(g^{(i)}_1) + \alpha_2 f_2(p) & = \beta_1 f_1(b)  + \beta_2 f_2(h^{(i)}_1) \\
    \alpha_1 f_1(g^{(i+1)}_1) + \alpha_2 f_2(p) & = \beta_1 f_1(b)  + \beta_2 f_2(h^{(i+1)}_1) 
  \end{aligned}
\end{equation*}

First two equations imply $\alpha_1 (f_1(g^{(i)}_1) - f_1(g^{(i+1)}_1)) = \alpha_1 (f_1(g^{(i+1)}_1) - f_1(g^{(i+2)}_1))$. The following two imply
$\beta_1(f_2(h^{(i)}_1) - f_2(h^{(i+1)}_1)) = \beta_1 (f_2(h^{(i+1)}_1) - f_2(h^{(i+2)}_1))$. Finally, the last two equations imply $\alpha_1 (f_1(g^{(i)}_1) -
f_1(g^{(i+1)}_1)) = \beta_1(f_2(h^{(i)}_1) - f_2(h^{(i+1)}_1))$. Apparently $\alpha_1 (f_1(g^{(i+1)}_1) - f_1(g^{(i+2}_1)) < \beta_1(f_2(h^{(i+1)}_1) - f_2(h^{(i+2)}_1))$ is
then a contradiction.
\linebreak

If we were to add an essentiality condition to Lemma \ref{lm:A5}, the statement can be made stronger as shown below. 
\paragraph{A5.}
Assume $ap \cleq bq, cp \cgeq dq$ and $ay_0 \sim x_0\pi(a), by_0 \sim x_0\pi(b), cy_1 \sim x_1\pi(c), dy_1 \sim x_1\pi(d)$, and also
$e\pi(a) \cgeq g\pi(b)$. Also, $X_1$ is essential on the set ($\NW{}$ or $\SE{}$) which includes $ap,bq,cp,dq$, and $X_2$ is essential on
the set ($\NW{}$ or $\SE{}$), which includes $x_0\pi(a)$ and $x_0\pi(b)$. 
Finally, assume $e\pi(c) \clt g\pi(d)$.

We get 
\begin{equation*}
  \begin{aligned}
    \alpha_1 f_1(a) + \alpha_2 f_2(p) & \leq \alpha_1 f_1(b) + \alpha_2 f_2(q) \\
    \alpha_1 f_1(c) + \alpha_2 f_2(p) & \geq \alpha_1 f_1(d) + \alpha_2 f_2(q) \\
\\
    \beta_1 f_1(e) + \beta_2 f_2(\pi(a)) & \geq \beta_1 f_1(g) + \beta_2 f_2(\pi(b)) \\
    \beta_1 f_1(e) + \beta_2 f_2(\pi(b)) & < \beta_1 f_1(g) + \beta_2 f_2(\pi(d)) \\
\\
    \gamma_1 f_1(a) + \gamma_2 f_2(y_0) & = \delta_1 f_1(x_0) + \delta_2 f_2(\pi(a)) \\
    \gamma_1 f_1(b) + \gamma_2 f_2(y_0) & = \delta_1 f_1(x_0) + \delta_2 f_2(\pi(b)) \\
\\
    \gamma_1 f_1(c) + \gamma_2 f_2(y_1) & = \delta_1 f_1(x_1) + \delta_2 f_2(\pi(c)) \\
    \gamma_1 f_1(d) + \gamma_2 f_2(y_1) & = \delta_1 f_1(x_1) + \delta_2 f_2(\pi(d))
  \end{aligned}
\end{equation*}

First two inequalities and the essentiality of $X_1$ ($\alpha_1 \neq 0$) imply $f_1(a) - f_1(b) \leq f_1(c) - f_1(d)$. Second pair of inequalities yields
$f_2(\pi(c)) - f_2(\pi(d)) < f_2(\pi(a)) - f_2(\pi(b))$, while the final pair of equations leads to $\gamma_1(f_1(c) - f_1(d)) =
\delta_2(f_2(\pi(c)) - f_2(\pi(d)))$.  Combining these results and due to essentiality of $X_2$ (hence $\delta_2 \neq 0$) we get:
\begin{equation*}
  \gamma_1(f_1(a) - f_1(b)) \leq \gamma_1(f_1(c) - f_1(d)) = \delta_2(f_2(\pi(c)) - f_2(\pi(d))) < \delta_2(f_2(\pi(a)) - f_2(\pi(b))),
\end{equation*}
which contradicts the third pair of inequalities above, which yield  $\gamma_1(f_1(a) - f_1(b)) = \delta_2(f_2(\pi(a)) -
f_2(\pi(b)))$.

\bibliographystyle{elsarticle-num-names}
\bibliography{cite_lib}

\end{document}